\documentclass[journal]{IEEEtran}
\usepackage{amsthm}
\usepackage{amssymb}
\usepackage{graphicx}
\usepackage{amsmath}
\usepackage{bm}
{}
{}
\newtheorem{lemma}{Lemma}{}
\newtheorem{proposition}{Proposition}{}
{}

\usepackage{algorithm}
\usepackage{algpseudocode}
\usepackage{graphics}
\usepackage{epsfig}
\usepackage{multirow}
\usepackage{array}
\usepackage{cite}
\usepackage{stfloats}
\usepackage{subfigure}
\usepackage{booktabs}
\usepackage{color}
\usepackage{tabu} 
\usepackage{stfloats}
\usepackage{mathtools}
\usepackage{stmaryrd}

\DeclareMathOperator*{\minimize}{minimize}

\begin{document}
\title{Integrated Sensing and Communication with Massive MIMO: A Unified Tensor Approach for Channel and Target Parameter Estimation}

\author{\IEEEauthorblockN{
		Ruoyu Zhang,
		Lei Cheng,
		Shuai Wang,
		Yi Lou,
		Yulong Gao, Wen Wu, \\
		and Derrick Wing Kwan Ng,~\IEEEmembership{Fellow,~IEEE} }
	\thanks{
		R. Zhang and W. Wu are with the Key Laboratory of Near-Range RF Sensing ICs \& Microsystems (NJUST), Ministry of Education, School of Electronic and Optical Engineering, Nanjing University of Science and Technology, Nanjing 210094, China (e-mail: ryzhang19@njust.edu.cn; wuwen@njust.edu.cn).
		L. Cheng is with the College of Information Science and Electronic Engineering, Zhejiang University, Hangzhou 310027, China (e-mail: lei\_cheng@zju.edu.cn).
		S. Wang is with the Shenzhen Institute of Advanced Technology, Chinese Academy of Sciences, Shenzhen 518055, China (e-mail: s.wang@siat.ac.cn).
		Y. Lou is with the College of Information Science and Engineering, Harbin Institute of Technology Weihai, Weihai 264209, China (e-mail: louyi@ieee.org).
		Y. Gao is with the School of Electronic and Information Engineering, Harbin Institute of Technology, Harbin 150001, China (e-mail: ylgao@hit.edu.cn).
		D. W. K. Ng is with the School of Electrical Engineering and Telecommunications, University of New South Wales, Sydney, NSW 2052, Australia (e-mail: w.k.ng@unsw.edu.au). }
\vspace{-2em}
}


\maketitle

\IEEEpeerreviewmaketitle

\begin{abstract}
Benefitting from the vast spatial degrees of freedom, the amalgamation of integrated sensing and communication (ISAC) and massive multiple-input multiple-output (MIMO) is expected to simultaneously improve spectral and energy efficiencies as well as the sensing capability. 
However, a large number of antennas deployed in massive MIMO-ISAC raises critical challenges in acquiring both accurate channel state information and target parameter information.
To overcome these two challenges with a unified framework, we first analyze their underlying system models and then propose a novel tensor-based approach that addresses both the channel estimation and target sensing problems.  
Specifically, by parameterizing the high-dimensional communication channel exploiting a small number of physical parameters, we associate the channel state information with the sensing parameters of targets in terms of angular, delay, and Doppler dimensions.
Then, we propose a shared training pattern adopting the same time-frequency resources such that both the channel estimation and target parameter estimation can be formulated as a canonical polyadic decomposition problem with a similar mathematical expression. 
On this basis, we first investigate the uniqueness condition of the tensor factorization and the maximum number of resolvable targets by utilizing the specific Vandermonde structure. 
Then, we develop a unified tensor-based algorithm to estimate the parameters including angles, time delays, Doppler shifts, and reflection/path coefficients of the targets/channels. 
In addition, we propose a segment-based shared training pattern to facilitate the channel and target parameter estimation for the case with significant beam squint effects. 
Simulation results verify our theoretical analysis and the superiority of the proposed unified algorithms in terms of estimation accuracy, sensing resolution, and training overhead reduction.
\end{abstract}


\section{Introduction}

Integrated sensing and communication (ISAC), which aims at unifying two functionalities into a single system to improve spectral and energy efficiencies, has been envisioned as a promising component for realizing the sixth-generation (6G) wireless communications \cite{Cui2021Integrating,Liu2022Integrated,Zhang2022Enabling}.  
Meanwhile, massive multiple-input multiple-output (MIMO), which advocates a large antenna array at the base station (BS), has become one of the most important enabling technologies for current and future wireless cellular networks \cite{larsson2014massive,Busari2018Millimeter}. 
Profiting from the vast spatial degrees of freedom and the endogenous sensing capability, 
the integration of massive MIMO and ISAC (MIMO-ISAC) is expected to simultaneously guarantee high-quality wireless communication as well as high-resolution and robust sensing,  
thus enabling the rapid development of various emerging applications including autonomous driving in intelligent transportation and unmanned aerial vehicle networks in smart cities \cite{Ali2020Leveraging,Zhang2021Perceptive}, etc.

Due to the different ways of information processing in sensing and communication systems, 
existing works for ISAC attempt to integrate both sensing and communication functionalities such that the information extraction from observations and the information delivery to receivers can be simultaneously realized \cite{Sturm2011Waveform,zhang2019multibeam,Zhang2022Integrated,Liyanaarachchi2021Optimized}.
To achieve this goal, 
in \cite{zhang2019multibeam}, the authors proposed a multibeam framework to simultaneously generate a fixed subbeam for communication and a varying scanning subbeam for sensing.
Alternatively, the orthogonal frequency-division multiplexing (OFDM) signal, a commonly adopted communication waveform, has been employed for target sensing \cite{Liyanaarachchi2021Optimized}. 
On the other hand, recent ISAC works developed a joint design methodology that can strike a scalable tradeoff between the dual functionalities via exploiting the extra flexibility and the degrees of freedom offered by MIMO techniques \cite{Liu2018Toward,liuxiang2020joint}. 
While the dual communication and sensing functionalities have been integrated into one platform to a certain extent, the aforementioned ISAC techniques generally assume a-priori of accurate communication channel state information or partially known target parameter information.  

Conventionally, the target parameter estimation problem has been independently studied in radar sensing systems. 
As one of key issues, direction-of-arrival (DOA) estimation problems were deeply studied over the past decades and various celebrated techniques have been proposed to provide super-resolution and high-accuracy estimation performance, such as multiple signal classification (MUSIC) \cite{Stoica1989MUSIC}, estimating signal parameters via rotational invariance techniques (ESPRIT) \cite{Roy1989ESPRIT}, and compressive sensing-based methods \cite{yang2013off,Cheng2015Subspace,Zhang2022DOA}. 
These existing techniques can be readily applied and extended to single-antenna ISAC systems \cite{liu2020super,xie2021performance,Zhang2020Joint} as well as MIMO-ISAC systems \cite{keskin2021mimo,rahman2019framework}. 
For example, an auto-paired range and velocity estimation method was proposed in \cite{liu2020super}, which makes use of the translational invariance structure of OFDM signals in both the pulse and frequency domains. 
Also, in \cite{xie2021performance}, a two-dimensional MUSIC algorithm along with the properly designed smoothing window was proposed to reduce the computational burden.
To further improve the estimation accuracy, the authors in \cite{Zhang2020Joint} considered the intrapulse and intersubcarrier Doppler effects and developed a maximum likelihood-based estimation scheme.
Subsequently, the authors in \cite{keskin2021mimo} proposed to estimate the delay, Doppler, and angle parameters of multiple targets for OFDM-based MIMO-ISAC systems. 
Besides, the authors in \cite{rahman2019framework} proposed to integrate the radar sensing function into the mobile communication network and developed the direct/indirect schemes for azimuth, range, and velocity estimation based on the one-dimensional compressive sensing algorithm. 
Although the aformentioned works can effectively estimate the target parameters adopting OFDM signals, the communication transmission requires the availability of accurate channel state information. In general, obtaining such a piece of information would consume exceedingly large amount of training resources in massive MIMO-ISAC systems that deploy a large number of antennas at the BS.

In practice, accurate channel estimation is essential for achieving the communication capacity \cite{Myers2019Message,Zhang2020Downlink} as well as reaping the potential performance gains of massive MIMO-ISAC. 
Over the years, there have been various approaches exploiting the channel structures of massive MIMO communication systems for reducing the channel training overhead \cite{Rod2018Frequency,Liao2019Closed,Zhang2022MMV,qin2018time,Liu2020Uplink,Wang2018Spatial,Wang2019An,Wang2020ABlock}.
For instance, the authors in \cite{Liao2019Closed} proposed a closed-loop framework for OFDM-based massive MIMO systems and a multi-dimensional ESPRIT algorithm to estimate the wideband millimeter-wave channels. 
Also, in \cite{Liu2020Uplink}, an uplink-aided downlink channel estimation scheme was  developed to estimate the time-varying MIMO channels, where the parameters, such as Doppler frequency, angles of departure (AoDs), and time delays, are jointly estimated via a variational Bayesian framework. 
To capture the non-negligible propagation delay when employing large array aperture and wide bandwidths, i.e., the so-called beam squint effects, the authors in \cite{Wang2020ABlock} designed a compressive sensing-based channel estimation algorithm that exploits the angle-delay sparsity to jointly compute the off-grid angles of arrival (AoAs) of the MIMO channel. 
Recently, tensor-based signal processing is attractive to exploit the multidimensional characteristics of massive MIMO channels and have been successfully applied for channel estimation problems such that the parameters of each channel multipath can be accurately estimated \cite{zhou2017low,park2019spatial,qian2019algebraic,lin2020tensor,Zhang2022Tensor}.
In \cite{qian2019algebraic}, with the help of carefully designing the training sequences, a canonical polyadic decomposition (CPD)-based channel estimation scheme  with low computational complexity was proposed for frequency division duplex-based massive MIMO systems.
Also, the authors in \cite{lin2020tensor} exploited the successive beam training strategy and developed a structured CPD-based channel estimation scheme to cope with the dual-wideband effects in massive MIMO-OFDM systems.

It should be highlighted that the aforementioned parameterized channel estimation approaches exhibit some similarities to the estimation of target angles and other parameters in traditional MIMO radar sensing systems.
The sensing process essentially acquires the target parameters by transmitting a known probing signal, which is similar to that of the channel estimation process and the targets can thus be virtually regarded as an uncooperative communication user \cite{liu2020radar,liu2020joint}.  
Inspired by the insights above, the authors in \cite{liu2020joint} proposed to jointly search the target and estimate the communication channel for millimeter wave MIMO systems, where the pilot signals are omnidirectionally transmitted by a hybrid analog-digital structure to actively detect the targets.
As a further step, the authors in \cite{yuan2021integrated} leveraged the topology of vehicular networks to estimate the kinematic parameters of vehicles, e.g., locations and speeds, such that the downlink beamformers can be constructed according to the predicted parameters without the explicit need for performing channel estimation.
Subsequently, a MIMO radar-aided channel estimation scheme was proposed to reduce the training overhead and improve the channel estimation accuracy in vehicle-to-everything scenarios \cite{huang2021mimo}. 
In practical massive MIMO-ISAC scenarios, the movement of targets and users, as well as the large bandwidth required for high performance communication and sensing, can lead to a substantially increase of unknown channel and target parameters. 
Although the training overhead for channel estimation can be significantly reduced with the assistance of the sensing echo signals, the aforementioned works \cite{yuan2021integrated,huang2021mimo} are only applicable for acquiring line-of-sight channels.
Note that various state-of-the-art channel estimation approaches have designed for the multipath channel estimation in different propagation environments, such as quasi-static \cite{Rod2018Frequency,Liao2019Closed,Zhang2022MMV,zhou2017low,qian2019algebraic,lin2020tensor}, high-mobility scenarios \cite{qin2018time,Liu2020Uplink,park2019spatial,Zhang2022Tensor}, and wideband systems with beam squint effects \cite{Wang2018Spatial,Wang2019An,Wang2020ABlock}.
However, these approaches can only acquire part of channel parameters in the delay, Doppler, and angular dimensions for specific scenarios, while fail to simultaneously estimate the range, azimuth, and velocity of the target parameters via the same training signals.

In this paper, we investigate the channel and target parameter estimation problem under a unified framework for massive MIMO-ISAC systems. 
The main contributions of the paper are summarized as follows:

\begin{itemize}
	\item 
	We make the first attempt to characterize the target parameter information and channel state information in a unified framework from the perspective of sensing \& communication channels. The framework is suitable for time-varying channels, frequency selective channels, as well as the case of the beam squint effects. 
	In particular, the high-dimensional massive MIMO channel can be parameterized by a small number of physical parameters, e.g., multipath delay, Doppler shift, and AoAs/AoDs, which establish an intrinsic link to the range, velocity, and azimuth angles, of the sensing targets.

	\item 
	 We propose a shared training pattern that exploits the same time-frequency resources for channel training and target sensing such that both the channel estimation and target parameter estimation can be formulated as two structured tensor decomposition problems, respectively. 
	 We model the received echo signals as a third-order canonical polyadic tensor such that the intrinsic multi-dimensional structure of the high-dimensional communication and sensing channels can be captured from the perspective of angular, delay, and Doppler dimensions. 
	 Besides, we analyze the maximum number of resolvable targets based on the uniqueness condition of the formulated tensor decomposition problem.
	
	\item 
	By exploiting the Vandermonde structure of the factor matrix and the spatial smoothing method, we propose a unified channel and target parameter estimation algorithm, where the AoA, AoD, time delay, Doppler shift, and coefficients can be estimated in a separate manner. 
	In addition, we propose an iterative estimation method for addressing the complicated coupling between AoD and Doppler shift parameters, providing the improved angle and velocity estimation accuracy of the targets, as well as the enhanced channel estimation performance. 

	\item 
	For the case of significant beam squint effects, we further propose a segment-based training pattern to facilitate the tensor formulation of both the channel estimation and target parameter estimation. 
	In the newly formulated tensor decomposition problem, we develop an effective parameters estimation scheme that not only exploits the Vandermonde structure but also avoids the coupling issue between angle and Doppler parameters. Besides, we analyze the corresponding uniqueness condition, guaranteeing that the channel and target parameters can be uniquely acquired from the observed signals.

\end{itemize}

Notations:  
Vectors and matrices are denoted by lowercase and uppercase boldface letters, respectively. 
$\left( \cdot \right)^{T}$, 
$\left( \cdot \right)^{*}$, 
$\left( \cdot \right)^{H}$, 
and $\left( \cdot \right)^{\dagger}$
represent the transpose, conjugate, conjugate transpose, and Moore-Penrose pseudo-inverse operations, respectively.
$\mathbf{I}_{M}$ is the $M \times M$ identity matrix and $\mathrm{j} = \sqrt{-1}$ denotes an imaginary unit. 
$\left\| \cdot \right\|_{2}$ and $\left\| \cdot \right\|_{F}$ denote the $l_{2}$ norm and Frobenius norm, respectively. 
$\odot$, $\otimes$, and $\circ$ denote the Khatri-Rao product, Kronecker product, and outer product, respectively. 
The Kruskal-rank and the rank of $\mathbf{A}$ are denoted by $k_{\mathbf{A}}$ and $\mathrm{rank}(\mathbf{A})$, respectively.
$\mathcal{R}\{\cdot\}$ extracts the real part of a complex number, while $\mathcal{CN}(0, \sigma^2)$ denotes the zero mean circularly-symmetric complex Gaussian distribution with the variance of $\sigma^2$. 
$[\mathbf{a}]_{m}$ and $[\mathbf{A}]_{m,n}$ denote the $m$-th entry and $(m,n)$-th entry of $\mathbf{a}$ and $\mathbf{A}$, respectively. 
$\mathrm{diag}(\mathbf{a})$ denotes a diagonal matrix formed by $\mathbf{a}$ and  $\mathbf{A}^{(K_1,1)} = [\mathbf{A}^{(1)}]_{1:K_1,:}$ denotes the submatrix that extracts the first $K_1$ rows of $\mathbf{A}^{(1)}$.

\section{Massive MIMO-ISAC System Model}

\begin{figure}[!tp]
	\centering
	\includegraphics[width=3.5in]
	{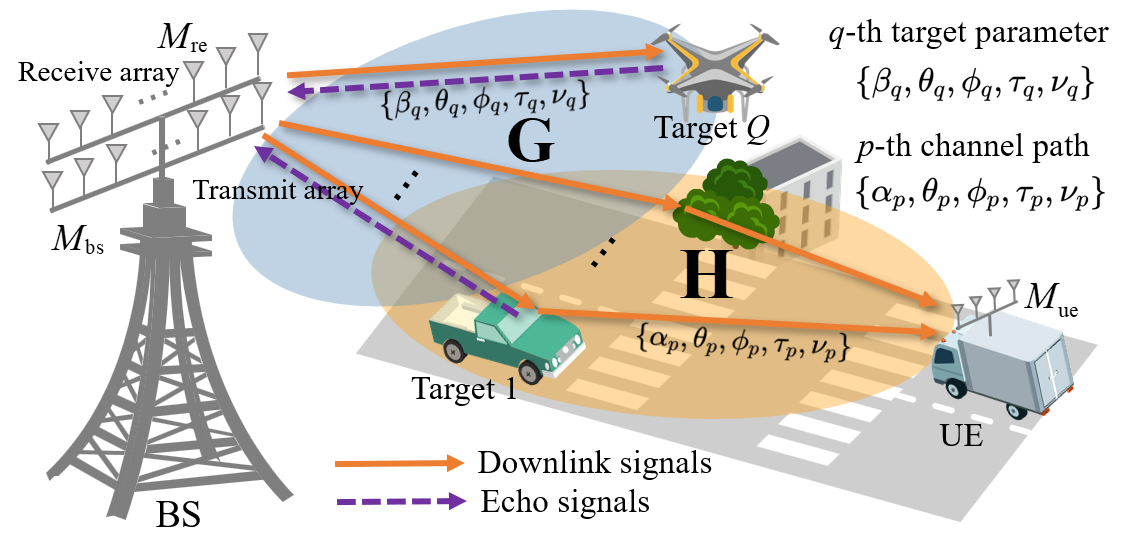}
	\caption{Massive MIMO-ISAC system model and the associated  communication/sensing channel models.}
	\label{Fig_SystemModel}
\end{figure}

As shown in Fig. \ref{Fig_SystemModel}, we consider a generic massive MIMO-ISAC system employing $M_{\mathrm{bs}}$ transmit antennas at the base station (BS). 
The BS transmits the shared wireless signals for serving an $M_{\mathrm{ue}}$-antenna communication user equipment (UE) and simultaneously sensing $Q$ surrounding targets.
The echo signals reflected by the targets are collected by spatially well-separated $M_{\mathrm{re}}$ receive antennas equipped at the BS, avoiding the self-interference caused by the transmitted signal leakage \cite{rahman2019framework}. 
Let $\tilde{s}(t) = e^{j2\pi f_c t} s(t)$ be the up-converted signal of the baseband signal $s(t)$ and the continuous-time passband signal for the transmission can be expressed as $\mathcal{R}\{ \tilde{s}(t) \}$. 
In this work, we focus on the downlink communication/sensing with a carrier frequency $f_{c}$,
where the UE and the BS individually perform the channel estimation and the target parameter estimation, respectively.\footnote{The extension to multi-user scenarios is straightforward for the downlink case, since the common downlink pilot signal is broadcasted from the BS and multiple UEs can perform channel estimation indepedently. For the uplink case, orthogonal time-frequency resources should be allocated to different users. For instance, denote $\mathcal{I}_{u}$ as the subcarrier index set for uplink channel training of the $u$-th UE, one can assign non-overlapping subcarrier index sets for different users, ensuring that channel estimation for multiple users can be carried out at the base station.}

\subsection{Communication Transmission Model}

Suppose a passband signal $\mathcal{R}\{ e^{j2\pi f_c t} s_{m_{\mathrm{bs}}}(t) \}$ of the $m_{\mathrm{bs}}$-th BS antenna arrives at the UE via $P$ multipaths, the received continuous-time signal at the $m_{\mathrm{ue}}$-th antenna can be given by \cite{Ertel1998Overview}
\begin{align}
	\tilde{r}_{m_{\mathrm{ue}}}\!(t) \!=\!\!\! \sum_{p=1}^P \!\mathcal{R}\!\! \left\{\! \tilde{\alpha}_{p} s_{m_{\mathrm{bs}}} \! (t \!-\!\! \tau_{p, m_{\mathrm{bs}}, m_{\mathrm{ue}}} \!) e^{j2\pi (f_c+\nu_{p}\!)\!\big(t-\tau_{p, m_{\mathrm{bs}}, m_{\mathrm{ue}}}\big)} \!\!\right\}\!\!,
\end{align}
where $\tilde{\alpha}_p \in \mathbb R$ is the passband path gain; 
$\nu_{p} \in [-\nu_{\max}/2,\nu_{\max}/2)$ denotes the Doppler shift;
$\nu_{\max}$ is termed as the (two-sided) Doppler spread of the channel; 
$\tau_{p, m_{\mathrm{bs}}, m_{\mathrm{ue}}}$ represents the time delay from the $m_{\mathrm{bs}}$-th BS antenna to the $m_{\mathrm{ue}}$-th antenna and is written as
\begin{align}
	\tau_{p, m_{\mathrm{bs}}, m_{\mathrm{ue}}} 
	\!= \tau_{p} 
	\!+\! \frac{(m_{\mathrm{bs}}\!-\!1) d \sin \phi_{p}}{c} 
	\!+\! \frac{(m_{\mathrm{ue}}\!-\!1) d \sin \theta_{p}}{c}, 
\end{align}
where $\tau_{p}$ is the free-space propagation delay between the first transmit and receive antenna, $\theta_{p}$ and $\phi_{p}$ denote the physical AoA and AoD of the $p$-th path, respectively; 
$d$ and $c$ denote the adjacent antenna spacing and the speed of light, respectively. 
Removing the carrier frequency $e^{j2\pi f_ct}$, the baseband received signal at the UE is
\begin{align}\label{yt}
	r_{m_{\mathrm{ue}}}(t) \!&=\! \sum_{p=1}^{P} \!\tilde{\alpha}_{p} s_{m_{\mathrm{bs}}}\!\big(t \!-\! \tau_{p, m_{\mathrm{bs}}, m_{\mathrm{ue}}}\big) e^{j2\pi \nu_{p} t} \notag\\
	& \quad\quad\quad\quad\quad\quad \times e^{- j2\pi (f_c+\nu_{p}) \tau_{p, m_{\mathrm{bs}}, m_{\mathrm{ue}}} }  \notag\\
	&=\!\! \sum_{p=1}^{P} \alpha_{p} s_{m_{\mathrm{bs}}}\!\big(t-\tau_{p, m_{\mathrm{bs}}, m_{\mathrm{ue}}} \big) e^{j2\pi \nu_{p}t}  \nonumber\\
	& \quad\quad\quad\quad\quad\quad \times
	e^{-j2\pi (f_c+\nu_{p}) (\Delta{\tau}_{m_{\mathrm{bs}}, p} + \Delta{\tau}_{m_{\mathrm{ue}}, p}) },
\end{align}
where $\alpha_p = \tilde{\alpha}_{p}e^{-j2\pi (f_c+\nu_{p})\tau_{p}}$ is the equivalent baseband path gain, $\Delta{\tau}_{m_{\mathrm{bs}}, p} \triangleq (m_{\mathrm{bs}}-1) d \sin \phi_{p}/c$ and $\Delta{\tau}_{m_{\mathrm{ue}}, p} \triangleq (m_{\mathrm{ue}}-1) d \sin \theta_{p}/c$ are the time delays introduced by the corresponding transmit and receive antennas, respectively. 
Then, the channel response between the $m_{\mathrm{ue}}$-th UE antenna and the $m_{\mathrm{bs}}$-th BS antenna can be represented as
\begin{align}
& [\mathbf{H}(t,\tau)]_{m_{\mathrm{ue}}, m_{\mathrm{bs}}} 
= \sum_{p=1}^{P}\alpha_{p}  e^{-j2\pi (f_c+\nu_{p})(\Delta{\tau}_{m_{\mathrm{bs}}, p} + \Delta{\tau}_{m_{\mathrm{ue}}, p}) }  \nonumber\\
&\quad\quad\quad\quad\quad\quad\quad\quad\quad \times e^{j2\pi \nu_{p}t}\delta(\tau - \tau_{p, m_{\mathrm{bs}}, m_{\mathrm{ue}}} ). \label{equ:H_t_tau}
\end{align}

By conducting the Fourier transform of \eqref{equ:H_t_tau} over $\tau$, we obtain the frequency-domain time-varying channel as
\begin{align}
	[\mathbf{H}(t,f)]_{m_{\mathrm{ue}}, m_{\mathrm{bs}}} 
	\!&\!\! = \sum_{p=1}^{P} \alpha_{p}  e^{-j2\pi (f + f_c + \nu_{p})  (\Delta{\tau}_{m_{\mathrm{bs}}, p} + \Delta{\tau}_{m_{\mathrm{ue}}, p}) }  \nonumber\\
	&\quad\quad\quad\quad\quad\quad \times   e^{j2\pi \nu_{p}t} e^{-j2\pi f \tau_p} \notag\\
	\!&\!\! = \sum_{p=1}^{P} \alpha_{p} e^{j2\pi \nu_{p}t} e^{-j2\pi f \tau_p} e^{-j2\pi (f + f_c) \Delta{\tau}_{m_{\mathrm{bs}}, p} }  \nonumber\\
	&\quad\quad\quad\quad\quad\quad  \times 
	e^{-j2\pi (f + f_c) \Delta{\tau}_{m_{\mathrm{ue}}, p} }, 
\end{align}
where the last equality holds due to the fact $\nu_{p} \ll f_c$.
Stacking the elements from all the BS and UE antennas, we can express the equivalent communication channel as 
\begin{align}
\mathbf{H}(t,f)=\sum_{p=1}^{P}\alpha_{p} \mathbf{a}_{\mathrm{ue}}(\theta_p,f)  \mathbf{a}_{\mathrm{bs}}^{T}(\phi_p,f)
e^{-j2\pi f \tau_p} e^{j2\pi \nu_{p}t}, \label{equ:CommuChannel_FreqTime}
\end{align}
where the wideband array steering vectors $\mathbf{a}_{\mathrm{bs}}(\phi_p,f)$ and $\mathbf{a}_{\mathrm{ue}}(\theta_p,f)$ are given by
\begin{align}\label{steering-bs}
	\mathbf{a}_{\mathrm{bs}}(\theta_p,f)&=\Big[1, e^{-j2\pi (1+\frac{f}{f_c})\frac{d\sin\phi_{p}}{\lambda_{c}}},\ldots, \nonumber\\
	&\quad\quad\quad\quad  
	e^{-j2\pi (1+\frac{f}{f_c}) \frac{(M_{\mathrm{bs}}-1)d\sin\phi_{p}}{\lambda_{c}}} \Big]^T,
\end{align}
\begin{align}\label{steering-ue}
	\mathbf{a}_{\mathrm{ue}}(\theta_p,f)&=\Big[1, e^{-j2\pi (1+\frac{f}{f_c})\frac{d\sin\theta_{p}}{\lambda_{c}}},\ldots,\nonumber\\
	&\quad\quad\quad\quad  
	e^{-j2\pi (1+\frac{f}{f_c}) \frac{(M_{\mathrm{ue}}-1)d\sin\theta_{p}}{\lambda_{c}}} \Big]^T,
\end{align}
respectively, where $\lambda_{c} = \frac{c}{f_c}$ denotes the carrier wavelength.

\subsection{Radar Sensing Model}

In this subsection, the radar sensing model in massive MIMO-ISAC systems is presented in detail. 
Suppose that the $q$-th point target exists at a distance $R_{q}$ with a radial velocity $V_{q}$, $\forall q \in \{1, \ldots, Q\}$. 
At the $m_{\mathrm{re}}$-th receive antenna of the BS, the echo signal bounced from the $Q$ targets can be given by \cite{Tigrek2012OFDM}
\begin{align} 
\tilde{y}_{m_{\mathrm{re}}}(t) 
&\! =\! \sum_{q=1}^{Q} \tilde{\beta}_{q} \tilde{s}_{m_{\mathrm{bs}}}\bigg(\!t -\! \frac{2(R_{q}-V_{q}t)}{c} -\! \Delta{\tau}_{m_{\mathrm{bs}},m_{\mathrm{re}},q}  \bigg),
\end{align}
where $\tilde{s}_{m_{\mathrm{bs}}}(t) = e^{j2\pi f_c t} s_{m_{\mathrm{bs}}}(t)$ is the transmit signal from the $m_{\mathrm{bs}}$-th BS antenna,
$\tilde{\beta}_{q}$ is the reflection coefficient of the $q$-th the target (proportional to the radar cross section),
$\Delta{\tau}_{m_{\mathrm{bs}},m_{\mathrm{re}},q} = \frac{(m_{\mathrm{bs}}-1) d \sin \phi_{q}}{c} +  \frac{(m_{\mathrm{re}}-1) d \sin \theta_{q}}{c}$,  
$\phi_{q}$ and $\theta_{q}$ denote the azimuth angle relative to the transmit and receive array of the $q$-th target, respectively.\footnote{For the monostatic radar where the transmit array and receive array are colocated, the azimuth angle satisfies $\phi_{q} = \theta_{q}$.} 
At the receiver of the BS, the signal after down-converting (removing the carrier frequency $e^{j2\pi f_c t}$ ) can be expressed as
\begin{align} 
y_{m_{\mathrm{re}}}(t) 
& = \sum_{q=1}^{Q} \tilde{\beta}_{q} e^{- j2\pi f_c \big( \frac{2(R-V_{q}t)}{c} + \Delta{\tau}_{m_{\mathrm{bs}},m_{\mathrm{re}},q}  \big) }  \nonumber\\
&\quad\quad\quad\quad  \times 
s_{m_{\mathrm{bs}}}\bigg(t - \frac{2(R_{q}-V_{q}t)}{c} - \Delta{\tau}_{m_{\mathrm{bs}},m_{\mathrm{re}},q}  \bigg) \notag\\
& = \sum_{q=1}^{Q} \beta_{q} e^{j2\pi \nu_{q} t} e^{- j2\pi f_c \Delta{\tau}_{m_{\mathrm{bs}}},m_{\mathrm{re}},q} \nonumber\\
&\quad\quad\quad\quad  \times
s_{m_{\mathrm{bs}}}\big(t - \tau_{q} - \Delta{\tau}_{m_{\mathrm{bs}},m_{\mathrm{re}},q}  \big), \label{equ:echoSig2}
\end{align}
where $\beta_{q} = \tilde{\beta}_{q} e^{- j2\pi f_c \tau_{q}}$ is the complex-valued reflection coefficient of the $q$-th target; 
$\tau_{q} = \frac{2R_{q}}{c}$ and $\nu_{q} = f_c \frac{2V_{q}}{c}$ represent the round-trip time delay and the Doppler shift, respectively;
the second equality in \eqref{equ:echoSig2} exploits $s(t - \tau_{q}) \approx s\big(t - \tau_{q} + 2V_{q}t/c \big)$ due to the target velocity assumption $V_{q} \ll c$, allowing the constant time delay during the coherent processing interval \cite{Tigrek2012OFDM}. 
Accordingly, we can obtain the equivalent sensing channel model as
\begin{align}
	g_{m_{\mathrm{re}}, m_{\mathrm{bs}}}(t,\tau) 
	& = \sum_{q=1}^{Q}  \beta_{q}  e^{j2\pi \nu_{q} t} e^{- j2\pi f_c \Delta{\tau}_{m_{\mathrm{bs}}, m_{\mathrm{re}},q} }\nonumber\\
	&\quad\quad\quad\quad  \times
	\delta\big(\tau - \tau_{q} - \Delta{\tau}_{m_{\mathrm{bs}}, m_{\mathrm{re}},q} \big), \label{equ:G_t_tau}
\end{align}
where $g_{m_{\mathrm{re}}, m_{\mathrm{bs}}}(t,\tau) $ can be viewed as the linear time-varying filter and $\delta(\cdot)$ is the Dirac delta function. Taking the Fourier transform of \eqref{equ:G_t_tau}  and stacking the channel from all the transmit and receive antennas, the sensing channel can be expressed as
\begin{align}
\mathbf{G}(t,f)=\sum_{q=1}^{Q}\beta_{q}   \mathbf{a}_{\mathrm{re}}(\theta_q,f)  \mathbf{a}_{\mathrm{bs}}^{T}(\phi_q,f)
e^{-j2\pi f \tau_q} e^{j2\pi \nu_{q} t}, \label{equ:RadarChannel_FreqTime}
\end{align}
where $\mathbf{a}_{\mathrm{re}}(\theta_q,f) \in \mathbb{C}^{M_{\mathrm{re}} \times 1}$ is the array steering vector given by
\begin{align}\label{steering-re}
	\mathbf{a}_{\mathrm{re}}(\theta_q,f)&=\Big[1,e^{-j2\pi (1+\frac{f}{f_c})\frac{d\sin\theta_{q}}{\lambda_{c}}},\ldots,\nonumber\\
	&\quad\quad\quad\quad  
	e^{-j2\pi (1+\frac{f}{f_c}) \frac{(M_{\mathrm{re}}-1)d\sin\theta_{q}}{\lambda_{c}}} \Big]^T.
\end{align}

\subsection{Unified Model for Massive MIMO-ISAC}
From \eqref{equ:CommuChannel_FreqTime} and \eqref{equ:RadarChannel_FreqTime}, 
it can be interestingly observed that the communication channel and the sensing channel exhibit a similar mathematical expression, which lays the foundation for developing a unified algorithmic framework to estimate the unknown channel and target parameters. 
The similarities of communication and sensing channels in massive MIMO-ISAC are analyzed from the following aspects. 

Remark 1: 
From the communication theory viewpoint, instead of estimating the entire channel matrix $\mathbf{H}$, one can alternatively estimate the physical parameters representing the channel, i.e., $\lbrace \alpha_{p}, \theta_{p}, \phi_{p}, \tau_{p}, \nu_{p} \rbrace_{p=1}^{P}$. 
These channel parameters include angles, time delays, and Doppler shifts, which are completely consistent with the estimation of target parameters $\lbrace \beta_{q}, \theta_{q}, \phi_{q}, \tau_{q}, \nu_{q} \rbrace_{q=1}^{Q}$ \cite{Zhang2021Perceptive}. 
The similar intrinsic structure provides a premise for channel and target parameter estimation exploiting the same signals, and thus is attractive to reduce the training overhead as well as improve the spectral efficiency.

Remark 2: 
Both the communication/sensing models in \eqref{equ:CommuChannel_FreqTime} and \eqref{equ:RadarChannel_FreqTime} characterize the scattering characteristics of channels/targets via a similar mathematical expression, but the parameters inside $\mathbf{H}$ and $\mathbf{G}$ exhibit different physical meanings. 
For example, $\alpha_{p}$ measures the scattering property in the wireless propagation environment, while $\beta_{q}$ describes the scattering intensity of the sensing targets.
Besides, $\tau_{p}$ and $\nu_{p}$ stand for the one-way time delay and the Doppler shift associated with the communication propagation path, 
while $\tau_{q}$ and $\nu_{q}$ denote the round-trip time delay and the Doppler shift caused by the target, respectively. 

Remark 3: 
In massive MIMO-ISAC with wide signal bandwidth, the physical propagation delay across the array aperture cannot be ignored\footnote{For massive MIMO-ISAC with 64 BS antennas, system bandwidth $f_{s} = 1$ GHz, $f_c = 28$ GHz, $\theta = 60^{\circ}$, the signal delay from the first antenna to the last one is $0.9743 T_{s}$, which is non-negligible compared to the sampling period $T_{s}$.} 
and the signals received by the antenna array experiences extra time delays compared to the observation of the first antenna. 
Such phenomenon is referred to as the beam squint effect, i.e., the steering vectors \eqref{steering-bs}, \eqref{steering-ue}, \eqref{steering-re} become frequency-dependent even with the same signal incident direction \cite{Wang2018Spatial}.
This challenging problem has been mainly studied for massive MIMO channel estimation with a single communication functionality, e.g., \cite{Wang2019An,Wang2020ABlock}, but rarely considered in the target parameter estimation problem.

In the following, we will show how to utilize the same training signals to effectively perform both the channel estimation and target parameter estimation for massive MIMO-ISAC via a unified framework.

\section{Unified Tensor Formulation and Analysis for Channel and Target Parameter Estimation}
In this section, we first develop a unified tensor formulation for the problem of target sensing and channel state information acquisition. Then, we conduct the uniqueness condition analysis of uniquely determining the target and channel parameters with the observed signals.

\subsection{Unified Tensor Modeling for Channel and Target Parameter Estimation}
We first illustrate how to formulate the target parameter estimation as a tensor model in massive MIMO-ISAC for the case with negligible beam squint effects. The sensing channel matrix at the $k$-th subcarrier of the $n$-th OFDM symbol can be expressed as 
\begin{equation} \label{equ:SensingChannelModel_k_n}
	\mathbf{G}_{n,k} = \sum_{q=1}^{Q} \beta_{q} e^{-\mathrm{j} 2 \pi \tau_{q} f_{s} \frac{k}{K_{0}}} \mathbf{a}_{\mathrm{re}}(\theta_{q})  \mathbf{a}_{\mathrm{bs}}^{T}(\phi_{q}) e^{\mathrm{j} 2 \pi  n \nu_{q} T_{\mathrm{sym}} },
\end{equation}
where $\mathbf{G}_{n,k}$ equals to  $\mathbf{G}(t,f)|_{t = n T_{\mathrm{sym}}, f = k \Delta f} $, $\mathbf{a}_{\mathrm{re}}(\theta_{q}) = \mathbf{a}_{\mathrm{re}}(\theta_{q}, f)|_{f=0}$, 
the steering vectors  $\mathbf{a}_{\mathrm{bs}}(\phi_{q}) = \mathbf{a}_{\mathrm{bs}}(\phi_{q}, f)|_{f=0}$, 
$K_{0}$, $f_{s}$, and $ \Delta f$ are the total number of subcarriers, the bandwidth for the transmission, and the subcarrier spacing, respectively,  
$T_{\mathrm{sym}} = T_{\mathrm{eff}} + T_{\mathrm{cp}}$, 
$T_{\mathrm{eff}}$ and $T_{\mathrm{cp}}$ denote the duration of the effective OFDM symbol and cyclic prefix (CP), respectively. 
Let $\mathbf{s}_{n,k} \in \mathbb{C}^{M_{\mathrm{bs}} \times 1}$ be the transmitted training symbol, $\forall k \in \{1, 2, \ldots, K\}$, $\forall n \in \{1, 2, \ldots, N\}$, where $N$ and $K$ denote the number of symbols and subcarriers for training, respectively. 
The symbol blocks are transformed into the time domain exploiting $M_{\mathrm{bs}}$ parallel $K_{0}$-point inverse discrete Fourier transform. 
The CP is employed to suppress the intersymbol interference caused by both the maximum time delay of multipath and the time delay corresponding to the maximum detectable range.
After the CP removal and the discrete Fourier transform at the BS, the echo signal can be given by
\begin{align}
	\overline{\mathbf{y}}_{n,k} = \mathbf{G}_{n,k}  \mathbf{s}_{n,k} + \overline{\mathbf{n}}_{n,k}, \label{equ:ReceivedSignalOriginal} 
\end{align}
where $\mathbf{G}_{n,k} \in \mathbb{C}^{M_{\mathrm{re}} \times M_{\mathrm{bs}}}$ represents the sensing channel matrix, 
$\overline{\mathbf{n}}_{n,k} \sim \mathcal{CN}(\mathbf{0}, \sigma^2 \mathbf{I}_{M_{\mathrm{re}}})$ is the additive noise vector. 

Note that the echo signals of \eqref{equ:ReceivedSignalOriginal} with multiple OFDM symbols and subcarriers cannot be independently decomposed into different dimensions, and thus does not naturally conform to the typical tensor model \cite{Cheng2017Probabilistic}.
To facilitate both the channel and the target parameter estimation, we propose a {\it shared training pattern} with precoded training structure, i.e., $\mathbf{s}_{n,k} = \mathbf{p}_{n} x_{k}$, where $\mathbf{p}_{n} \in \mathbb{C}^{M_{\mathrm{bs}} \times 1}$ is the training precoder and $x_{k} \in \mathbb{C}$ is the pilot symbol satisfying $|x_{k}|^2 = 1$. 
Note that $\mathbf{p}_{n}$ varies for different OFDM symbols while $x_{k}$ remains unchanged during the training stage.
Also, the proposed training structure will subsequently assist to formulate the echo signals as a third-order tensor. 
After removing the pilot symbol at the receiver, the echo signals can be further expressed as
\begin{align}
& \mathbf{y}_{n,k} 
= x_{k}^{*} \overline{\mathbf{y}}_{n,k} \nonumber\\
& = \sum_{q=1}^{Q} \beta_{q} e^{-\mathrm{j} 2 \pi \tau_{q} f_{s} \frac{k}{K_{0}}} \mathbf{a}_{\mathrm{re}}(\theta_{q})  \mathbf{a}_{\mathrm{bs}}^{T}(\phi_{q}) \mathbf{p}_{n} e^{\mathrm{j} 2 \pi  n \nu_{q} T_{\mathrm{sym}} } + \mathbf{n}_{n,k}, \label{equ:ReceivedSignalVector}
\end{align}
where $\mathbf{n}_{n,k} = x_{k}^{*} \overline{\mathbf{n}}_{n,k}$ is the corresponding noise vector. 
Concatenating the echo signals of the $N$ training OFDM symbols yields
\begin{align}
	\mathbf{Y}_{k} 
	& =\! \sum_{q=1}^{Q} \beta_{q} e^{-\mathrm{j} 2 \pi \tau_{q} f_{s} \frac{k}{K_{0}}} \mathbf{a}_{\mathrm{re}}(\theta_{q})  \mathbf{a}_{\mathrm{bs}}^{T}(\phi_{q})  
	\mathbf{P} \mathbf{\Gamma}(\nu_{q}) + \mathbf{N}_{k},  \label{equ:ReceivedSignalMatrix}
\end{align}
where  
$\mathbf{Y}_{k} = [ \mathbf{y}_{1,k}, \mathbf{y}_{2,k}, \ldots, \mathbf{y}_{N, k}] \in \mathbb{C}^{M_{\mathrm{re}} \times N }$ contains the overall echo signals at the $k$-th subcarrier; 
$\mathbf{\Gamma}(\nu_{q}) \in \mathbb{C}^{N \times N }$ is a diagonal matrix with the $n$-th diagonal element being $e^{\mathrm{j} 2 \pi  n \nu_{q} T_{\mathrm{sym}} }$;  
$\mathbf{P} = [\mathbf{p}_{1}, \ldots, \mathbf{p}_{N} ] \in \mathbb{C}^{M_{\mathrm{bs}} \times N }$ is the training precoder matrix and $\mathbf{N}_{k} = [\mathbf{n}_{1}, \ldots, \mathbf{n}_{N}] \in \mathbb{C}^{M_{\mathrm{re}} \times N }$ denotes the noise matrix.

By collecting the echo signals of $K$ training subcarriers, we can reformulate the echo signals as a third-order tensor $\bm{\mathcal{Y}} \in \mathbb{C}^{M_\mathrm{re} \times N \times K}$ that admits the following CPD format \cite{Cheng2017Probabilistic}
\begin{align}
	\bm{\mathcal{Y}} & = \sum_{q=1}^{Q} \mathbf{a}_{\mathrm{re}}(\theta_{q}) \circ \mathbf{b}_{\mathrm{bs}}(\phi_{q}, \nu_{q}) 
	\circ \mathbf{c}(\beta_{q}, \tau_{q}) + \bm{\mathcal{N}},  \label{equ:CPD Model}
\end{align}
where $\bm{\mathcal{N}} \in \mathbb{C}^{M_\mathrm{re} \times N \times K}$ is the noise tensor, 
$\mathbf{b}_{\mathrm{bs}}(\phi_{q}, \nu_{q}) = \mathbf{\Gamma}(\nu_{q})  \mathbf{P}^{T} \mathbf{a}_{\mathrm{bs}}(\phi_{q})$, and $\mathbf{c}(\beta_{q}, \tau_{q}) = \beta_{q} \mathbf{a}_{\mathrm{td}}(\tau_{q}) $ with
\begin{align} \label{equ:a_tau}
	\mathbf{a}_{\mathrm{td}}(\tau_{q}) = \Big[e^{-\mathrm{j} 2 \pi \tau_{q} f_{s} \frac{1}{K_{0}}}, \ldots, e^{-\mathrm{j} 2 \pi \tau_{q} f_{s} \frac{K}{K_{0}}} \Big]^{T}.
\end{align}

Note that our goal is to estimate the unknown target parameters $\{\theta_{q}, \phi_{q}, \nu_{q}, \tau_{q}, \beta_{q} \}_{q=1}^{Q}$ and the target parameter estimation problem can be formulated as 
\begin{equation} \label{equ:TargetParameterEstimation}
	\minimize_{ \{\theta_{q}, \: \phi_{q}, \: \nu_{q}, \atop \: \tau_{q}, \: \beta_{q} \}_{q=1}^{Q} } \left\|\bm{\mathcal{Y}} - \! \sum_{q=1}^{Q}  \mathbf{a}_{\mathrm{re}}(\theta_{q})  \circ  \mathbf{b}_{\mathrm{bs}}(\phi_{q}, \nu_{q}) 
	\circ \mathbf{c}(\beta_{q}, \tau_{q}) \right\|_{F}^{2}\!\!\!.
\end{equation}

From \eqref{equ:TargetParameterEstimation}, we can observe that the unknown parameters $\{ \theta_{q}, \phi_{q}, \tau_{q}, \beta_{q}, \nu_{q} \}$ are nonlinearly coupled together and it is challenging to perform the joint optimization directly. 
Fortunately, aided by the uniqueness condition of CPD \cite{zhou2017low,park2019spatial,qian2019algebraic}, an effective strategy for the parameters estimation is to exploit the factor matrices of \eqref{equ:CPD Model} and they are given by
\begin{align}
	\mathbf{B}^{(1)} & \triangleq \big[ \mathbf{a}_{\mathrm{re}}(\theta_{1}), \ldots, \mathbf{a}_{\mathrm{re}}(\theta_{Q}) \big] 
	\in \mathbb{C}^{M_\mathrm{re} \times Q},  \label{equ:factor_matrix_1} \\
	\mathbf{B}^{(2)} & \triangleq \big[ \mathbf{b}_{\mathrm{bs}}(\phi_{1}, \nu_{1}), \ldots, \mathbf{b}_{\mathrm{bs}}(\phi_{Q}, \nu_{Q}) \big]
	\in \mathbb{C}^{N \times Q}, \label{equ:factor_matrix_2} \\
	\mathbf{B}^{(3)} & \triangleq \big[ \mathbf{c}(\beta_{1}, \tau_{1}), \ldots, \mathbf{c}(\beta_{Q}, \tau_{Q}) \big] 
	\in \mathbb{C}^{K \times Q},  \label{equ:factor_matrix_3}
\end{align}
which incorporate the unknown parameters to be estimated and are associated with the observations along the space, time, and frequency domains, respectively.

Meanwhile, the received training signals at the UE side are employed for channel estimation. Thanks to the proposed shared training pattern, the channel estimation problem can be represented as a third-order tensor via the similar modeling method as the echo signals given by \eqref{equ:SensingChannelModel_k_n}-\eqref{equ:factor_matrix_3}. 
Consequently, the channel estimation problem can be similarly formulated as:
\begin{equation} \label{equ:ChannelParameterEstimation}
	\minimize_{ \{\theta_{p}, \: \phi_{p}, \: \nu_{p}, \atop \: \tau_{p}, \: \alpha_{p} \}_{p=1}^{P} } \left\|\bm{\mathcal{R}} - \! \sum_{p=1}^{P}  \mathbf{a}_{\mathrm{ue}}(\theta_{p}) \! \circ \! \mathbf{b}_{\mathrm{bs}}(\phi_{p}, \nu_{p}) 
	\! \circ \! \mathbf{c}(\alpha_{p}, \tau_{p}) \right\|_{F}^{2}\!\!\!,
\end{equation}
where $\bm{\mathcal{R}} \in \mathbb{C}^{M_\mathrm{ue} \times N \times K}$ is the observation tensor at the UE side for channel estimation. 
In the following, we discuss the uniqueness condition for the identification of these parameters.

\subsection{Uniqueness Property Analysis for Unified Tensor Formulation}

The uniqueness condition of the CPD problem is essential for channel and target parameter estimation, ensuring that the decomposed factor matrices incorporate the entire information of unknown parameters.
The uniqueness of CPD can be guaranteed by Kruskal’s condition \cite{Stegeman2007On}, which is given as follows. 

\begin{lemma} \label{lemma:Kruskal}
	Let $\mathbf{B}^{(1)} \in \mathbb{C}^{I_{1} \times Q}$, $\mathbf{B}^{(2)}\in \mathbb{C}^{I_{2} \times Q}$, and $\mathbf{B}^{(3)}\in \mathbb{C}^{I_{3} \times Q}$ be the factor matrices of $\bm{\mathcal{X}} \in \mathbb{C}^{I_{1} \times I_{2} \times I_{3}}$.
	Suppose the Kruskal’s condition
	$
		k_{\mathbf{B}^{(1)}} + k_{\mathbf{B}^{(2)}} + k_{\mathbf{B}^{(3)}} \geq 2Q + 2
	$ holds,
	then the CPD of $\bm{\mathcal{X}}$ is said to be unique. 
	In the generic case\footnote{The generic case means that the entries of the factor matrices are drawn from absolutely continuous probability density functions with probability one \cite{qian2019algebraic}.}, the uniqueness condition becomes
	\begin{align} \label{equ:UniquenessCondition}
		\min(I_{1},Q) +  \min(I_{2},Q) + \min(I_{3}, Q) \geq 2Q + 2.
	\end{align}
\end{lemma}
\begin{proof}
	Please refer to \cite{Stegeman2007On} for the detailed proof.
\end{proof}

Lemma \ref{lemma:Kruskal} provides a sufficient condition to determine the uniqueness of the CPD problem. Alternatively, given the number of receive antennas, OFDM symbols, and training subcarriers, the number of targets $Q$ that can be recovered is determined. 
This value can be further enlarged after considering the algebraic structures of the factor matrices, e.g., Vandemonde structure \cite{Srensen2013Blind}. 
In the considered target parameter estimation \eqref{equ:TargetParameterEstimation} (or channel estimation \eqref{equ:ChannelParameterEstimation}) problem, we find that $\mathbf{a}_{\mathrm{td}}(\tau_{q})$ can be viewed as a steering vector and thus the factor matrix $\mathbf{B}^{(3)}$ exhibits a Vandermonde structure,
whose generators are $\big\{ z_{\tau,q} = e^{- \mathrm{j} 2 \pi \tau_{q} f_{s} / K} \big\}_{q=1}^{Q}$. 
With this structural information, one can obtain the following relaxed uniqueness condition.

\begin{lemma} \label{lemma:Lemma2_Vander}
Let $\mathbf{B}^{(1)} \in \mathbb{C}^{I_{1} \times Q}$, $\mathbf{B}^{(2)}\in \mathbb{C}^{I_{2} \times Q}$, and $\mathbf{B}^{(3)}\in \mathbb{C}^{I_{3} \times Q}$ be the factor matrices of $\bm{\mathcal{X}} \in \mathbb{C}^{I_{1} \times I_{2} \times I_{3}}$,  
where $\mathbf{B}^{(3)}$ exhibits the Vandermonde structure with distinct generators.  
If 
\begin{align}
\begin{cases}
\mathrm{rank} \big(\mathbf{B}^{(K_3-1, 3)} \odot \mathbf{B}^{(2)} \big) = Q, \\
\mathrm{rank} \big(\mathbf{B}^{(L_3, 3)} \odot \mathbf{B}^{(1)} \big) = Q, \label{equ:Lemma2_Vander}
\end{cases}
\end{align}
with $K_{3} + L_{3} = I_{3} + 1$, the CPD of $\bm{\mathcal{X}}$ is said to be unique. 
In the generic case, \eqref{equ:Lemma2_Vander} becomes
\begin{align}
	\min \big( I_{2} (K_{3}-1),  I_{1}L_{3} \big) \geq Q. \label{equ:theoUniqueness}
\end{align}
\end{lemma}
\begin{proof}
The proof similarly follows that of Theorem III.3 in \cite{Srensen2013Blind} by taking the mode-1 unfolding of $\bm{\mathcal{Y}}$ in the noiseless setting, i.e., 
$\mathbf{Y}_{(1)}^{T} = \big( \mathbf{B}^{(3)} \odot \mathbf{B}^{(2)} \big) \mathbf{B}^{(1)T}$.
\end{proof}

\renewcommand\arraystretch{1.4}
\begin{table}[!tp]
	\centering
	\caption{The maximum number of resolvable targets under the uniqueness conditions given in Lemma \ref{lemma:Kruskal} and Lemma \ref{lemma:Lemma2_Vander}, where $K = 16$, $M_\mathrm{re} = 8$.} 
	\begin{tabular}{|c|c|c|c|c|c|c|}
		\hline
		$N$ & 10 & 12 & 14 & 16 & 18 & 20 \\
		\hline
		Lemma \ref{lemma:Kruskal} & 
		16 & 17 & 18 & 19 & 20 & 21 \\
		\hline
		Lemma \ref{lemma:Lemma2_Vander}, $K_3 = 3$ & 
		20 & 24 & 28 & 32 & 36 & 40 \\
		\hline
		Lemma \ref{lemma:Lemma2_Vander}, $K_3 = 5$ 
		& 40 & 48 & 56 & 64 & 72 & 80 \\
		\hline
		Lemma \ref{lemma:Lemma2_Vander}, $K_3 = 7$ 
		& 60 & 72 & 80 & 80 & 80 & 80 \\
		\hline
	\end{tabular}
	\label{Table:max_num_tar}
\end{table}

Lemma \ref{lemma:Lemma2_Vander} shows a stronger identifiability result for the Vandermonde constrained CPD problems and reveals the following two practical insights to the parameter estimation in massive MIMO-ISAC systems.
First, the uniqueness condition in Lemma \ref{lemma:Lemma2_Vander} holds true for the circumstances that the factor matrices have dependent columns. 
This allows reliable parameters estimation even when different targets are located closely spaced in direction, providing higher spatial resolution for target sensing.  
Second, given the dimensions of factor matrices, Lemma \ref{lemma:Lemma2_Vander} indicates that it is possible to deal with the CPD problem with higher tensor rank than that of traditional CPD models without structural constraints \cite{kolda2009tensor}. 
This means that the maximum number of resolvable targets can be improved when considering the Vandermonde constraints. 
Table \ref{Table:max_num_tar} compares the upper bound of the maximal target number $Q$ under different uniqueness conditions and it can be seen that Lemma \ref{lemma:Lemma2_Vander} yields the relaxed uniqueness conditions.

\section{Unified Channel and Target Parameter Estimation Algorithm}

In this section, we propose a unified parameter estimation algorithm that consists of two phases.
In the first phase, we explain the details of how to estimate the factor matrices in \eqref{equ:TargetParameterEstimation}. 
In the second phase, we illustrate how to leverage the estimated factor matrices to extract the unknown parameters.
Besides, we discuss the practical issues for implementing the proposed algorithm.

\subsection{Estimation of Factor Matrices}

To estimate the factor matrices from the tensor $\bm{\mathcal{Y}}$, one of the most well-known algorithms is the alternating least squares (ALS)-based method \cite{kolda2009tensor}. 
Although the ALS-based method is efficient to acquire the estimation of factor matrices, it does not take the structural constraints of factor matrices into account and generally limits the maximum number of resolvable targets. 
To address this issue, we employ the Vandermonde structure in $\mathbf{B}^{(3)}$ and propose an improved parameter estimation method based on Lemma \ref{lemma:Lemma2_Vander}.
To be specific, taking the matrix unfolding of $\bm{\mathcal{Y}}$ along its first dimension and the transpose operation, we have
\begin{align}
	\mathbf{Y}_{(1)}^{T} = ( \mathbf{B}^{(3)} \odot \mathbf{B}^{(2)} ) \mathbf{B}^{(1)T} + \mathbf{N}_{(1)}^{T}, \label{equ:Mode1Unfolding}
\end{align}
where $\mathbf{N}_{(1)} \in \mathbb{C}^{M_{\mathrm{re}} \times K N}$ is the corresponding noise matrix.
Due to the Vandermonde structure of $\mathbf{B}^{(3)}$, the spatial smoothing technique can be applied to expand the dimension of $\mathbf{Y}_{(1)}^{T}$. 
We construct a cyclic selection matrix
$\mathbf{J}_{l_3} = \big[\mathbf{0}_{K_3 \times l_3 - 1}  \ \mathbf{I}_{K_3} \ \mathbf{0}_{K_3 \times L_3 - l_3} \big] \otimes \mathbf{I}_{N}  \in \mathbb{C}^{K_3 N  \times K N}$ and vary $l_3$ from 1 to $L_3$, yielding
\begin{align}
	\mathbf{Y}_{S} & = \big[
	\mathbf{J}_{1} \mathbf{Y}_{(1)}^{T} \ \mathbf{J}_{2} \mathbf{Y}_{(1)}^{T} \ \ldots 
	\ \mathbf{J}_{L_3} \mathbf{Y}_{(1)}^{T}
	\big]  \nonumber\\
	& = \big( \mathbf{B}^{(K_3, 3)} \odot \mathbf{B}^{(2)} \big) \mathbf{\Lambda}
	\big( \mathbf{B}^{(L_3, 3)} \odot \mathbf{B}^{(1)} \big)^T + \mathbf{N}_{S},
\end{align}
where $K_3 + L_3 = K + 1$, the second equality exploits the Vandermonde property, $\mathbf{\Lambda}$ is a diagonal matrix, and $\mathbf{N}_{S} \in \mathbb{C}^{K_3 N \times L_3 M_\mathrm{re}}$ is the corresponding noise matrix.

Sequentially, one can follow \cite{Srensen2013Blind} and utilize the ESPRIT-like method to fulfil the factor matrices estimation. 
To be specific, compute the singular value decomposition
$\mathbf{Y}_{S} = \mathbf{U} \bm{\Sigma} \mathbf{V}^H = \mathbf{U}_{s} \bm{\Sigma}_{s} \mathbf{V}_{s}^H + \mathbf{U}_{n} \bm{\Sigma}_{n} \mathbf{V}_{n}^H$, where the $Q$ principal singular vectors in $\mathbf{U}$ and $\mathbf{V}$ span the signal subspaces $\mathbf{U}_{s}$ and $\mathbf{V}_{s}$, respectively, $\bm{\Sigma}_{s}$ is diagonal matrix formed by the $Q$ largest singular values. 
Define the submatrices $\mathbf{U}_{1} = [\mathbf{U}_{s}]_{1:(K_3-1)N,:}$ and
$\mathbf{U}_{2} = [\mathbf{U}_{s}]_{N + 1: K_3 N,:}$, 
and compute the eigenvalue decomposition as 
$\mathbf{U}_{1}^{\dagger} \mathbf{U}_{2}= \mathbf{M} \mathbf{Z} \mathbf{M}^H$. 
Accordingly, the $q$-th column of the factor matrix $\hat{\mathbf{B}}^{(3)}$ can be reconstructed as
\begin{align}
\hat{\mathbf{b}}_{q}^{(3)} = [\hat{z}_{\tau,q}, \hat{z}_{\tau,q}^2, \ldots, \hat{z}_{\tau,q}^{K}]^T, \label{equ:bb3_hat}
\end{align}
where $\hat{z}_{\tau,q} = [\mathbf{Z}]_{q,q}/\big| [\mathbf{Z}]_{q,q} \big|$ is the estimated generator and $[\mathbf{Z}]_{q,q}$ is $q$-th diagonal element of $\mathbf{Z}$.
Then, the $q$-th column of $\hat{\mathbf{B}}^{(2)}$ and $\hat{\mathbf{B}}^{(1)}$ can be constructed as
\begin{align}
\hat{\mathbf{b}}_{q}^{(2)} = \big( \hat{\mathbf{b}}_{q}^{(K_3, 3) H} \otimes \mathbf{I}_{N} \big) \mathbf{U}_{s} \mathbf{m}_{q}, \label{equ:bb2_hat}
\end{align}
\begin{align}
\hat{\mathbf{b}}_{q}^{(1)} = \big( \hat{\mathbf{b}}_{q}^{(L_3, 3) H} \otimes  \mathbf{I}_{M_{\mathrm{re}}} \big) \mathbf{V}_{s}^{*} \bm{\Sigma}_{s} \mathbf{t}_{q}, \label{equ:bb1_hat}
\end{align}
respectively, where $\mathbf{m}_{q}$ is the $q$-th column of $\mathbf{M}$ and 
$\mathbf{t}_{q}$ is the $q$-th column of $\mathbf{M}^{*}$. 
The estimated factor matrices in \eqref{equ:bb3_hat}-\eqref{equ:bb1_hat} and the true factor matrices satisfy the relation
\begin{align} \label{equ:FactorMatrixEst}
	\hat{\mathbf{B}}^{(\kappa)} = \mathbf{B}^{(\kappa)} \mathbf{\Pi} \mathbf{\Delta}_{\kappa} + \mathbf{E}_{\kappa},  \quad  \kappa = 1,2,3,
\end{align}
where $\mathbf{\Delta}_{\kappa}$ is the unknown diagonal matrix satisfying $\mathbf{\Delta}_{1}\mathbf{\Delta}_{2}\mathbf{\Delta}_{3} = \mathbf{I}_{Q}$. 
$\mathbf{\Pi} \in \mathbb{C}^{Q \times Q}$ represents a permutation matrix and 
$\mathbf{E}_{\kappa}$ denotes the errors of tensor factorization.

\subsection{Extraction of Unknown Parameters}
\label{subsec:TVFS Channel Parameters}

In this subsection, we turn to the extraction of unknown parameters  $\theta_{q}$, $\phi_{q}$, $\tau_{q}$, $\beta_{q}$, and $\nu_{q}$ with the estimated $\lbrace \hat{\mathbf{B}}^{(1)}, \hat{\mathbf{B}}^{(2)}, \hat{\mathbf{B}}^{(3)} \rbrace$.

By observing the relationship of factor matrices in \eqref{equ:FactorMatrixEst}, we find that each column of $\mathbf{B}^{(1)}$ is associated one angle of arrival. 
Accordingly, we can separately extract the AoA $\theta_{q}$ via the correlation-based method with respect to each column of $\hat{\mathbf{B}}^{(1)}$ \cite{zhou2017low}, i.e.,
\begin{align} \label{equ:AoA_est}
\hat{\theta}_{q} 
= \arg\max_{\theta} \Big| \big(\hat{\mathbf{b}}_{q}^{(1)}\big)^{H} \mathbf{a}_{\mathrm{re}}(\theta) \Big|^{2},
\end{align}
where $\hat{\mathbf{b}}_{q}^{(1)}$ is the $q$-th column of the estimated factor matrix $\hat{\mathbf{B}}^{(1)}$. 
The time delay $\{ \hat{\tau}_{q} \}$ can be directly obtained via the estimated generators $\{ \hat{z}_{\tau,q} \}$ of the factor matrix $\hat{\mathbf{B}}^{(3)}$, i.e.,
\begin{align}
\hat{\tau}_{q} = - \frac{K_0}{2\pi f_{s}}  \angle \hat{z}_{\tau,q}, \ q = 1, \ldots, Q,  \label{equ:tau_est}
\end{align}
where $\angle$ denotes the phase angle extraction operator.

The AoD and the Doppler shift are incorporated in the estimated factor matrix $\hat{\mathbf{B}}^{(2)}$ and the corresponding $q$-th column $\hat{\mathbf{b}}_{q}^{(2)}$ satisfies
\begin{equation} \label{equ:Sigma_a_theta}
	\hat{\mathbf{b}}_{q}^{(2)}
	= \delta_{q} \mathbf{\Gamma}(\nu_{q}) \mathbf{P}^{T} \mathbf{a}_{\mathrm{bs}}(\phi_{q}) + \mathbf{e}_{q},
\end{equation}
where $\delta_{q}$ and $\mathbf{e}_{q}$ are the $q$-th diagonal element and the $q$-th column of $\mathbf{\Delta}_{2}$ and $\mathbf{E}_{2}$, respectively. 
The joint estimation of the AoD $\phi_{q}$ and the Doppler shift $\nu_{q}$ from $\hat{\mathbf{b}}_{q}^{(2)}$ is given by 
\begin{align} 
	\lbrace \hat{\nu}_{q}, \hat{\phi}_{q} \rbrace = \arg\min_{\nu_{q}, \phi_{q}} 
	\Big\lVert \hat{\mathbf{b}}_{q}^{(2)}  - \delta_{q} \mathbf{\Gamma}(\nu_{q}) \mathbf{P}^{T} \mathbf{a}_{\mathrm{bs}}(\phi_{q}) \Big\rVert_{2}^2. \label{equ:JointAoDandDopplerProblem}
\end{align}

As can be observed from \eqref{equ:JointAoDandDopplerProblem}, three unknown parameters $\delta_{q}$, $\nu_{q}$, and $\phi_{q}$ are coupled together, which is very challenging to jointly estimate these three parameters.
Besides, since the unknown Doppler shift $\nu_{q}$ introduces extra phase offsets, conventional correlation-based methods that ignore the phase offsets will cause severe performance degradation. 
To deal with the aforementioned difficulties, we propose an alternative optimation that can iteratively refine the estimate of AoD $\hat{\phi}_{q}$ and the Doppler shift $\hat{\nu}_{q}$, $\forall q$.

Given $\hat{\nu}_{q}^{[i-1]}$ as the Doppler shift estimate at the $(i-1)$-th iteration, the AoD estimation at the $i$-th iteration is formulated as the following problem
\begin{align} 
\lbrace \hat{\phi}_{q}^{[i]} \rbrace 
& = \arg\min_{\phi_{q}} 
\Big\lVert \hat{\mathbf{b}}_{q}^{(2)}  - \delta_{q} \mathbf{\Gamma}(\hat{\nu}_{q}^{[i-1]}) \mathbf{P}^{T} \mathbf{a}_{\mathrm{bs}}(\phi_{q}) \Big\rVert_{2}^2 \nonumber\\
& = \arg\min_{\phi_{q}} 
\Big\lVert \hat{\mathbf{d}}_{q}^{[i-1]}  - \delta_{q} \mathbf{P}^{T} \mathbf{a}_{\mathrm{bs}}(\phi_{q}) \Big\rVert_{2}^2, \label{equ:AoDandDoppler_AoD}
\end{align}
where the second equality comes from $\hat{\mathbf{d}}_{q}^{[i-1]} = \mathbf{\Gamma}^{-1}\big(\hat{\nu}_{q}^{[i-1]} \big) \hat{\mathbf{b}}_{q}^{(2)}$. 
On the other hand, given the estimated AoD $\hat{\phi}_{q}^{[i]}$, the estimation of Doppler shifts at the current iteration is expressed as
\begin{align} 
	\lbrace \hat{\nu}_{q}^{[i]} \rbrace & = \arg\min_{\nu_{q}} 
	\Big\lVert \hat{\mathbf{b}}_{q}^{(2)}  - \delta_{q} \mathbf{\Gamma}(\nu_{q}) \mathbf{P}^{T} \mathbf{a}_{\mathrm{bs}}(\hat{\phi}_{q}^{[i]}) \Big\rVert_{2}^2  \nonumber\\
	& = \arg\min_{\nu_{q}} 
	\Big\lVert \hat{\mathbf{b}}_{q}^{(2)}  - \delta_{q} \mathbf{Q}^{T} \mathbf{a}_{\mathrm{do}}(\nu_{q}) \Big\rVert_{2}^2, \label{equ:AoDandDoppler_Doppler}
\end{align}
where the second equality follows from defining a diagonal matrix $\mathbf{Q}$, whose diagonal elements are $\mathbf{P}^{T} \mathbf{a}_{\mathrm{bs}}\big(\hat{\phi}_{q}^{[i]}\big)$, and extracting the diagonal elements of $\mathbf{\Gamma}(\nu_{q})$ as
\begin{align} \label{equ:a_doppler}
	\mathbf{a}_{\mathrm{do}}(\nu_{q}) = \Big[e^{\mathrm{j} 2 \pi \nu_{q} T_{\mathrm{sym}} }, e^{\mathrm{j} 4 \pi \nu_{q} T_{\mathrm{sym}} }, \ldots, e^{\mathrm{j} 2 \pi  N \nu_{q} T_{\mathrm{sym}} } \Big]^{T}.
\end{align}

Now, the joint optimization problem in \eqref{equ:JointAoDandDopplerProblem} can be recast as two individual estimation problems for each variable. 
Desipte the solution of \eqref{equ:AoDandDoppler_AoD} and \eqref{equ:AoDandDoppler_Doppler} can be obtained via a one-dimensional search, 
we can leverage the Vandermonde structure of $\mathbf{a}_{\mathrm{bs}}(\phi)$ and $\mathbf{a}_{\mathrm{do}}(\nu)$ and derive the closed-form solution by a polynomial method \cite{park2019spatial}. 
To be specific, for \eqref{equ:AoDandDoppler_AoD}, denoting the steering vector $\mathbf{a}_{\mathrm{bs}}(\phi)$ as $\mathbf{a}(z_{\phi}) = \big[1, z_{\phi}, \ldots, z_{\phi}^{M_{\mathrm{bs}}-1} \big]^T$ with $\big\{ z_{\phi} = e^{ - \mathrm{j} 2 \pi d \sin\phi / \lambda_{c}} \big\}$, 
we have
\begin{align}
	\hat{z}_{\phi} & = \arg\min_{z_{\phi}} \Bigg(1 - \frac{\big| \hat{\mathbf{d}}_{q}^{H} \mathbf{P}^{T} \mathbf{a}(z_{\phi}) \big|^{2} }{\big\| \hat{\mathbf{d}}_{q}\big\|_{2}^{2} \big\|\mathbf{P}^{T} \mathbf{a}(z_{\phi})\big\|_{2}^{2}}  \Bigg)    \nonumber\\
	& = \arg\min_{z_{\phi}} \Bigg( \frac{ \mathbf{a}^{H}(z_{\phi})   \mathbf{W}_{\mathbf{P},q}  \mathbf{a}(z_{\phi}) }{\big\| \hat{\mathbf{d}}_{q}\big\|_{2}^{2} \big\|\mathbf{P}^{T} \mathbf{a}(z_{\phi})\big\|_{2}^{2}}  \Bigg), \label{equ:AoD_est_roots}
\end{align}
where $\mathbf{W}_{\mathbf{P},q} =  \mathbf{P}^{*} \Big( \big\| \hat{\mathbf{d}}_{q}\big\|_{2}^{2} \mathbf{I}_{N} - \hat{\mathbf{d}}_{q} \big(\hat{\mathbf{d}}_{q}\big)^{H} \Big) \mathbf{P}^{T}$.   
The numerator in \eqref{equ:AoD_est_roots} can be further expressed as a polynomial with respect to $z_{\phi}$ and equals to zero in the noiseless case, yielding
\begin{align}
	\mathbf{a}^{H}(z_{\phi})   \mathbf{W}_{\mathbf{P},q}  \mathbf{a}(z_{\phi}) 
	= \!\!\!\!\!\!\!\! \sum_{m = -M_{\mathrm{bs}} + 1}^{M_{\mathrm{bs}}-1} \!\!\! \Bigg( \! \sum_{n_1 - n_2 = m}  \!\!\! [ \mathbf{W}_{\mathbf{P},q} ]_{n_1, n_2} \Bigg)  z_{\phi}^{m} = 0.
\end{align}
There are $2 (M_{\mathrm{bs}}-1)$ roots and half of them are located within the unit circle. 
Let $\{ \gamma_{m} \}_{m=1}^{M_{\mathrm{bs}}-1}$ be the roots with unit amplitudes within the unit circle. 
Then, we can acquire the solution $z_{\phi}^{\star}$ of \eqref{equ:AoD_est_roots} by searching over $z_{\phi} \in \{ \gamma_{m} \}_{m=1}^{M_{\mathrm{bs}}-1} $. 
Accordingly, the estimate of AoD can be given by
\begin{align}
\hat{\phi}_{q}^{[i]} = \arcsin \big( - \frac{\lambda_{c}}{2 \pi d} \angle z_{\phi}^{\star} \big). \label{equ:phi_closedform}
\end{align}

We next address the Doppler shift estimation problem in \eqref{equ:AoDandDoppler_Doppler}.
Due to the similar mathematical form, the solution of \eqref{equ:AoDandDoppler_Doppler} can be obtained in the same way as \eqref{equ:AoDandDoppler_AoD}.
Define the matrix 
$\mathbf{W}_{\mathbf{Q},q} = \mathbf{Q}^{*} \Big( \big\| \hat{\mathbf{b}}_{q}^{(2)}\big\|_{2}^{2} \mathbf{I}_{N} - \hat{\mathbf{b}}_{q}^{(2)} \big(\hat{\mathbf{b}}_{q}^{(2)}\big)^{H} \Big) \mathbf{Q}^{T}$, and the steering vector  $\mathbf{a}(z_{\nu}) = \big[z_{\nu}, z_{\nu}^2, \ldots, z_{\nu}^{T} \big]^T$ with the generator $\big\{ z_{\nu} = e^{\mathrm{j} 2 \pi \nu T_{\mathrm{sym}} } \big\}$,
we can obtain the estimation of the generator for Doppler shifts by replacing $\hat{\mathbf{d}}_{q}^{[i-1]}$, 
$\mathbf{P}$,
$\mathbf{W}_{\mathbf{P},q}$, 
and $\mathbf{a}(z_{\phi})$, 
with 
$\hat{\mathbf{b}}_{q}^{(2)}$, 
$\mathbf{Q}$,
$\mathbf{W}_{\mathbf{Q},q}$, 
and $\mathbf{a}(z_{\nu})$, 
respectively. 
Let $z_{\nu}^{\star}$ be the estimated generator, and the estimate of Doppler shift can be obtained by 
\begin{align}
\hat{\nu}_{q}^{[i]} = \frac{1}{2 \pi T_{\mathrm{sym}}} \angle z_{\nu}^{\star}. \label{equ:nu_closedform}
\end{align}

\begin{algorithm}[!tp]
	\renewcommand{\algorithmicrequire}{\textbf{Input:}}  
	\renewcommand{\algorithmicensure}{\textbf{Output:}}  
	\caption{Proposed Parameter Estimation Algorithm}
	\label{alg:ParaEst}
	\begin{algorithmic}[1]
		\Require
		The received echo signals $\mathbf{y}_{n,k}$, $\forall k \in \{1, 2, \ldots, K\}$, $\forall n \in \{1, 2, \ldots, N\}$, $\mathbf{P}$, 
		initial Doppler shift $\nu_{q}^{[0]} = 0$,
		the number of iterations $I_{\mathrm{iter}}$.
		\State Compute the factor matrices $\mathbf{B}^{(3)}$, $\mathbf{B}^{(2)}$, and $\mathbf{B}^{(1)}$ via \eqref{equ:bb3_hat}-\eqref{equ:bb1_hat}, respectively;
		
		\For {each $q \in [1,\ldots,Q]$}
		\State Estimate $\hat{\theta}_{q}$ via \eqref{equ:AoA_est}; \
		\State Estimate $\hat{\tau}_{q}$ via \eqref{equ:tau_est}; \
		\For {$i = 1,\ldots, I_{\mathrm{iter}}$}
		\State Estimate $\hat{\phi}_{q}^{[i]}$ and $\hat{\nu}_{q}^{[i]}$ via \eqref{equ:AoDandDoppler_AoD} and \eqref{equ:AoDandDoppler_Doppler}, respectively; \  
		\EndFor
		\State Return $\hat{\phi}_{q} = \phi_{q}^{[I_{\mathrm{iter}}]}$ and $\hat{\nu}_{q} = \nu_{q}^{[I_{\mathrm{iter}}]}$; \  
		\EndFor
		\State Estimate $\hat{\beta}_{q}$ via \eqref{equ:beta}; \  		
		\Ensure
		Return estimated parameters $\lbrace \hat{\theta}_{q}, \hat{\tau}_{q}, \hat{\phi}_{q}, \hat{\nu}_{q}, \hat{\beta}_{q} \rbrace_{q=1}^{Q}$.
	\end{algorithmic}
\end{algorithm}

Based on the above procedures, we alternately perform the estimation of $\hat{\phi}_{q}$ and $\hat{\nu}_{q}$ until the iterative update process converges to a stationary point of \eqref{equ:JointAoDandDopplerProblem}. 
After obtaining $\lbrace \hat{\nu}_{q} \rbrace_{q=1}^{Q}$, $\lbrace \hat{\phi}_{q} \rbrace_{q=1}^{Q}$, $\lbrace \hat{\theta}_{q} \rbrace_{q=1}^{Q}$, and $\lbrace \hat{\tau}_{q} \rbrace_{q=1}^{Q}$, we are able to estimate the reflection coefficients.
From the observation tensor model \eqref{equ:CPD Model}, taking the mode-3 unfolding of $\bm{\mathcal{Y}}$ and the transpose operation, we have
\begin{align}
	\mathrm{vec} \big(\mathbf{Y}_{(3)}^{T}\big) 
	& = \mathbf{A}_{\mathrm{td}} \odot  ( \mathbf{B}^{(2)} \odot \mathbf{B}^{(1)} ) \bm{\beta} + \mathrm{vec} \big(\mathbf{N}_{(3)}^{T}\big),
\end{align}
where $\bm{\beta} = [\beta_{1}, \ldots, \beta_{Q}]^T$ and
$\mathbf{a}_{\mathrm{td}}(\tau_{q})$ is the $q$-th column of $\mathbf{A}_{\mathrm{td}}$.
Following the least-squares criterion, the estimate of reflection coefficients can be obtained by 
\begin{align} 
	\hat{\bm{\beta}} = \big(\hat{\mathbf{A}}_{\mathrm{td}} \odot  ( \hat{\mathbf{B}}^{(2)} \odot \hat{\mathbf{B}}^{(1)} ) \big)^{\dagger} \mathrm{vec} \big(\mathbf{Y}_{(3)}^{T}\big), \label{equ:beta}
\end{align}
where $\hat{\mathbf{A}}_{\mathrm{td}}$, $\hat{\mathbf{B}}^{(2)}$, and $\hat{\mathbf{B}}^{(1)}$ are constructed based on the estimated parameters $\lbrace \hat{\tau}_{q} \rbrace_{q=1}^{Q}$,  $\lbrace \hat{\nu}_{q} \rbrace_{q=1}^{Q}$ and $\lbrace \hat{\phi}_{q} \rbrace_{q=1}^{Q}$, and  $\lbrace \hat{\theta}_{q} \rbrace_{q=1}^{Q}$, respectively.

To this end, the unknown parameters of the targets can be obtained at the BS side, which is summarized in Algorithm \ref{alg:ParaEst}.  
Meanwhile, the UE can estimate the five-tuple set of channel parameters $\lbrace \hat{\theta}_{p}, \hat{\tau}_{p}, \hat{\phi}_{p}, \hat{\nu}_{p}, \hat{\alpha}_{p} \rbrace_{p=1}^{P}$ under the same process.
With the estimated multipath channel parameters, one can reconstruct the entire channel matrix $\mathbf{\hat{H}}_{n,k}$ according to \eqref{equ:CommuChannel_FreqTime}.

\subsection{Discussions}
In this subsection, we will discuss the practical issues for implementing the proposed algorithm, including the realistic scenario with electromagnetic interference from non-target objects, unknown number of targets, pilot signals design, and the performance tradeoff between sensing and communication.

We first discuss a more realistic scenario, where there are electromagnetic interference from non-target objects or passive interference caused by the links between targets and UEs. 
From the perspective of receive signal processing, the echo signals comprising the interference can be expressed as
\begin{align}
	\mathbf{y}_{n,k} 
	& = \sum_{q=1}^{Q} \beta_{q} e^{-\mathrm{j} 2 \pi \tau_{q} f_{s} \frac{k}{K_{0}}} \mathbf{a}_{\mathrm{re}}(\theta_{q})  \mathbf{a}_{\mathrm{bs}}^{T}(\phi_{q}) \mathbf{p}_{n} e^{\mathrm{j} 2 \pi  n \nu_{q} T_{\mathrm{sym}} }  \nonumber\\
	& \quad + \underbrace{\sum_{\tilde{q}=1}^{\tilde{Q}} \beta_{\tilde{q}} e^{-\mathrm{j} 2 \pi \tau_{\tilde{q}} f_{s} \frac{k}{K_{0}}} \mathbf{a}_{\mathrm{re}}(\theta_{\tilde{q}})  \mathbf{a}_{\mathrm{bs}}^{T}(\phi_{\tilde{q}}) \mathbf{p}_{n} e^{\mathrm{j} 2 \pi  n \nu_{\tilde{q}} T_{\mathrm{sym}} }}_{\text{interference}} \nonumber\\
	& \quad\quad\quad\quad\quad + \mathbf{n}_{n,k}, 
\end{align}
where $\lbrace \beta_{\tilde{q}}, \theta_{\tilde{q}}, \phi_{\tilde{q}}, \tau_{\tilde{q}}, \nu_{\tilde{q}} \rbrace_{{\tilde{q}}=1}^{\tilde{Q}}$ denotes the unknown parameters of $\tilde{Q}$ non-target objects or the number of passive interference links. 
For this model, our proposed algorithm has two processing ways: 
i) treating the interference as the additional noise. In this case, the total noise power increases, the proposed algorithm will operate in lower signal-to-noise ratio (SNR) regime for estimating the true targets. 
ii) treating the interference as the virtual targets. In this case, the noise power remains unchanged, while the unknown parameters will incorporate both true targets and non-target objects, i.e., $\lbrace \beta_{q}, \theta_{q}, \phi_{q}, \tau_{q}, \nu_{q} \rbrace_{{\tilde{q}}=1}^{Q+\tilde{Q}}$.  
For the proposed algorithm, besides the true target, the estimated results will contain more objects and their corresponding parameters. Indeed, distinguishing and identifying these two types of targets and non-target objects remains a challenge that deserves for further investigation. We will leave this target classification and identification problem as a future research topic.

Remark 1: 
We assumed that the tensor rank, i.e., the number of targets $Q$, is known for the factor matrices estimation.  
If the knowledge of $Q$ is unavailable, one can employ more sophisticated tensor decomposition techniques or adopt a preliminary step to determine it. 
For the former, the basic idea is to exploit sparsity-promoting priors to find a low-rank representation of the echo signal tensor such that the tensor rank and the factor matrices can be jointly estimated \cite{Bazerque2013Rank,Cheng2022Towards}. 
For the latter, a preliminary step is adopted to determine the tensor rank based on certain criteria, such as the minimum description length \cite{Liu2016Detection,Yokota2017Robust}. In this way, the number of targets $Q$ can be estimated, and then the factor matrices estimation can be performed based on the estimated $Q$. 

Remark 2: 
For the proposed parameter estimation algorithm, the training pattern of pilot signals needs to have a random phase, ensuring that the training beams can measure the channels in an isotropic manner, instead of exhaustive directional beams. Accordingly, the required overhead is related to the number of targets or channel multipaths, rather than the number of transceiver antennas. On the other hand, the required channel estimation time for the proposed algorithm is proportional to the number of OFDM symbols employed for training. 
This is because the channel estimation can be carried out after collecting $N$ OFDM symbols in the time domain, enabling the received signals to be formulated as a third-order tensor. 

Remark 3: 
The integration of ISAC and massive MIMO is expected to design the sensing and communication functionalities in a unified manner, attempting to share the same spectrum resources and merge the dual functionalities into a single system. 
Note that the communication and sensing channels are usually different, especially the number of multipaths and the number of targets are not always equal to each other. 
This disparity will result in a fundamental performance tradeoff when exploiting the same time-frequency resources to estimate these two different channels. 
To be specific, when the number of multipaths of the communication channel is fewer than the number of targets, the proposed algorithm can attain better channel estimation performance. 
On the contrary, i.e., when there are fewer targets than multipaths, the proposed algorithm can achieve higher target parameter estimation accuracy.
On the other hand, in terms of energy optimization between these two tasks, it necessitates optimizing the signal energy based on the desired channel estimation or target parameter estimation performance. 
For instance, the signal energy should be optimized such that the task with weaker performance can satisfy the requirements of sensing or communication, then the performance of another task can also satisfy the requirements.

\section{Parameter Estimation With Beam Squint Effects}
In this section, we investigate the parameter estimation for the case with beam squint effects. In this case, the steering vectors in \eqref{equ:RadarChannel_FreqTime} become frequency-dependent and the sensing channel matrix can be accordingly expressed as 
\begin{equation} \label{equ:DW_SensingChannelModel_k_n}
	\mathbf{G}_{n,k} = \sum_{q=1}^{Q} \beta_{q} e^{-\mathrm{j} 2 \pi \tau_{q} f_{s} \frac{k}{K_{0}}} \mathbf{a}_{\mathrm{re},k}(\theta_{q})  \mathbf{a}_{\mathrm{bs},k}^{T}(\phi_{q}) e^{\mathrm{j} 2 \pi  n \nu_{q} T_{\mathrm{sym}} },
\end{equation}
where  $\mathbf{a}_{\mathrm{re},k}(\theta_{q}) = \mathbf{a}_{\mathrm{re}}(\theta_{q}, f)|_{f=k \Delta f}$, 
$\mathbf{a}_{\mathrm{bs},k}(\phi_{q}) = \mathbf{a}_{\mathrm{bs}}(\phi_{q}, f)|_{f=k \Delta f}$. 
Unfortunately, the observation data of \eqref{equ:ReceivedSignalMatrix} with multiple subcarriers cannot be independently decomposed into multiple dimensions, and thus does not conform to the tensor model. 
In the following, we propose a CPD tensor formulation and the corresponding parameter estimation algorithm by exploiting signals of multiple time slots.

\subsection{CPD Formulation With Beam Squint Effects}

To enable the CPD formulation of the received signals in the presence of beam squint effects, there is a need to design a new training pattern for facilitating the parameters estimation. 
To this end, the segment-based training precoder is designed as follows  
\begin{align} 
\mathbf{P} & = \big[ \mathbf{p}_{1}, \mathbf{p}_{2}, \ldots, \mathbf{p}_{N} \big] \nonumber \\
& = \big[ \underbrace{\overline{\mathbf{p}}_{1}, \ldots, \overline{\mathbf{p}}_{N_{d}}}_{(1)} \Big|
\underbrace{\overline{\mathbf{p}}_{1}, \ldots, \overline{\mathbf{p}}_{N_{d}}}_{(2)} \Big| 
\cdots \Big|
\underbrace{\overline{\mathbf{p}}_{1}, \ldots, \overline{\mathbf{p}}_{N_{d}}}_{(L)} \big]. \label{equ:training_BeamSquint}
\end{align}

As can be seen, the whole channel training frame in the time domain is divided into $L$ segments, each containing $N_{d}$ OFDM symbols such that $N_{d} L = N$.
In other words, the training precoder is updated $N_{d}$ times within one time slot and repeated $L$ times per frame. 
At the receiver, after removing the pilot symbol and concatenating the echo signals of $N_d$ training OFDM symbols, the echo signals of the $l$-th segment can be given by
\begin{align}
	\overline{\mathbf{Y}}_{k,l} 
	& =\! \sum_{q=1}^{Q} \beta_{q} e^{-\mathrm{j} 2 \pi \tau_{q} f_{s} \frac{k}{K_{0}}} \mathbf{a}_{\mathrm{re},k}(\theta_{q})  \mathbf{a}_{\mathrm{bs},k}^{T}(\phi_{q})  
	\overline{\mathbf{P}} \overline{\mathbf{\Gamma}}_{l}(\nu_{q}) + \overline{\mathbf{N}}_{k,l} \nonumber\\
	& =\! \sum_{q=1}^{Q} \overline{\beta}_{q,k}  \mathbf{a}_{\mathrm{re},k}(\theta_{q})  \mathbf{a}_{\mathrm{bs},k}^{T}(\phi_{q})  
	\overline{\mathbf{P}} \overline{\mathbf{\Gamma}}_{l}(\nu_{q}) + \overline{\mathbf{N}}_{k,l}  \label{equ:ReceivedSignalMatrix_DW}
\end{align}
where 
$\mathbf{\overline{Y}}_{k,l} \in \mathbb{C}^{ M_{\mathrm{re}} \times N_{d}}$, 
$\overline{\mathbf{P}} = [\overline{\mathbf{p}}_{1}, \ldots, \overline{\mathbf{p}}_{N_{d}}] \in \mathbb{C}^{ M_{\mathrm{bs}} \times N_{d}}$, 
$\overline{\mathbf{\Gamma}}_{l}(\nu_{q}) \in \mathbb{C}^{N_{d} \times N_{d} }$ is a diagonal matrix with the $n_{d}$-th diagonal element being $e^{\mathrm{j} 2 \pi ((l-1) N_{d}+n_{d}) \nu_{q} T_{\mathrm{sym}} }$, 
$\mathbf{\overline{N}}_{k,l}$ is the corresponding noise matrix, and the second equality holds by defining 
$\overline{\beta}_{q,k} = \beta_{q} e^{-\mathrm{j} 2 \pi \tau_{q} f_{s} \frac{k}{K_{0}}}$.

Since the moving speed of the target is far less than the speed of light \cite{zhang2019multibeam}, the phase induced by Doppler effects can be regarded as quasi-static and approximately unchanged within consecutive $N_{d}$ OFDM symbols,\footnote{The assumption is generally satisfied in typical ISAC application scenarios. For instance, in a wideband massive MIMO system with a carrier frequency of $f_{c} = 28$ GHz, $K_{0} = 128$ subcarriers, and $f_{s} = 1$ GHz, the OFDM symbol period $T_{\mathrm{sym}} = 0.16$ $\mu$s (assuming $T_{\mathrm{cp}} = K_{0}/(4f_{s})$). Considering a vehicle target moving at a relative speed at 120 km/h, the round-trip Doppler shift is $\nu_q = 6.22$ kHz. In this case, $\nu_q T_{\mathrm{sym}} < 0.001 \ll 1$, and the phase induced by the Doppler shift can be regarded as static over multiple OFDM symbols.} and we have the approximation
$e^{\mathrm{j} 2 \pi \nu_{q} t } \approx \mathrm{const}$  when $t \in  \big[(l-1) N_{d} T_{\mathrm{sym}}, l N_{d} T_{\mathrm{sym}} \big]$, $l = 1, 2, \ldots, L$.
Consequently, the echo signals in \eqref{equ:ReceivedSignalMatrix_DW} can be expressed as 
\begin{align} 
	\mathbf{\overline{Y}}_{k,l} 
	& \!=\! \sum_{q=1}^{Q} \overline{\beta}_{q,k} e^{\mathrm{j} 2 \pi  l N_{d} \nu_{q} T_{\mathrm{sym}} } \mathbf{a}_{\mathrm{re},k}(\theta_{q})  \mathbf{\overline{b}}_{\mathrm{bs},k}^{T}(\phi_{q}) \!+\! \widetilde{\mathbf{N}}_{k,l}, \label{equ:ReceivedSignalMatrix_N1_OFDM}
\end{align}
where 
$\mathbf{\overline{b}}_{\mathrm{bs},k}(\phi_{q}) = \mathbf{\overline{P}}^{T} \mathbf{a}_{\mathrm{bs},k}(\phi_{q}) \in \mathbb{C}^{ N_{d} \times 1}$, 
$\widetilde{\mathbf{N}}_{k,l}$ contains the propogation noise and the Doppler approximation noise. 
By collecting $\mathbf{\overline{Y}}_{k,l} $ of multiple time slots $\big\{\mathbf{\overline{Y}}_{k,l} \big\}_{l=1}^{L}$, we can formulate the received echo signals as a tensor $\bm{\overline{\mathcal{Y}}}_{k} \in \mathbb{C}^{M_\mathrm{re} \times N_{d} \times L}$ fitting the CPD format as
\begin{align}
\bm{\overline{\mathcal{Y}}}_{k} 
& = \sum_{q=1}^{Q}  \mathbf{a}_{\mathrm{re},k}(\theta_{q}) \circ  \mathbf{\overline{b}}_{\mathrm{bs},k}^{T}(\phi_{q}) \circ
\big(\overline{\beta}_{q,k} \mathbf{a}_{\mathrm{ds}}(\nu_{q}) \big) + \bm{\widetilde{\mathcal{N}}}_{k},  \label{equ:tensor_BeamSquint}
\end{align}
where $\mathbf{a}_{\mathrm{ds}}(\nu_{q}) = \big[e^{\mathrm{j} 2 \pi \nu_{q} N_{d} T_{\mathrm{sym}} }, \ldots, e^{\mathrm{j} 2 \pi \nu_{q} L N_{d}  T_{\mathrm{sym}} } \big] \in \mathbb{C}^{L \times 1}$. 
At the $k$-th subcarrier, 
the three factor matrices for the tensor $\bm{\overline{\mathcal{Y}}}_{k}$ are given by
\begin{align}
	\mathbf{C}_{k}^{(1)} & \triangleq \big[ \mathbf{a}_{\mathrm{re},k}(\theta_{1}), \ldots, \mathbf{a}_{\mathrm{re},k}(\theta_{Q}) \big] 
	\in \mathbb{C}^{M_\mathrm{re} \times Q},  \label{equ:DW_factor_matrix_1} \\
	\mathbf{C}_{k}^{(2)} & \triangleq \big[ \mathbf{\overline{b}}_{\mathrm{bs},k}(\phi_{1}), \ldots, \mathbf{\overline{b}}_{\mathrm{bs},k}(\phi_{Q}) \big]
	\in \mathbb{C}^{N_{d} \times Q}, \label{equ:DW_factor_matrix_2} \\
	\mathbf{C}_{k}^{(3)} & \triangleq \big[ \overline{\beta}_{1,k} \mathbf{a}_{\mathrm{ds}}(\nu_{1}), \ldots, \overline{\beta}_{Q,k} \mathbf{a}_{\mathrm{ds}}(\nu_{Q}) \big] 
	\in \mathbb{C}^{L \times Q}.  \label{equ:DW_factor_matrix_3}
\end{align}

\subsection{Uniqueness Condition}
In the observation model \eqref{equ:tensor_BeamSquint}, we find that  $\mathbf{a}_{\mathrm{ds}}(\nu_{q})$ can be regarded as a steering vector with respect to the Doppler shift and thus the factor matrix $\mathbf{C}_{k}^{(3)}$ exhibits a Vandermonde structure, whose generators are $\big\{ z_{\nu,q} = e^{\mathrm{j} 2 \pi \nu_{q} N_{d} T_{\mathrm{sym}} } \big\}_{q=1}^{Q}$. 
With the structural information, we follow Lemma \ref{lemma:Lemma2_Vander} and obtain the following uniqueness condition of CPD in the case of beam squints.

\begin{proposition} \label{prop:Uniqueness_BeamSquint}
	Let $\mathbf{C}_{k}^{(1)} \in \mathbb{C}^{M_\mathrm{re} \times Q}$, 
	$\mathbf{C}_{k}^{(2)} \in \mathbb{C}^{N_{d} \times Q}$, 
	$\mathbf{C}_{k}^{(3)} \in \mathbb{C}^{L \times Q}$ be the factor matrices of $\bm{\mathcal{X}}_{k} \in \mathbb{C}^{M_\mathrm{re} \times N_{d} \times L}$,  
	where $\mathbf{C}_{k}^{(3)}$ exhibits the Vandermonde nature with distinct generators.  
	If 
	\begin{align}
		\begin{cases}
			\mathrm{rank} \big(\mathbf{C}_{k}^{(K_3-1, 3)} \odot \mathbf{C}_{k}^{(2)} \big) = Q, \\
			\mathrm{rank} \big(\mathbf{C}_{k}^{(L_3, 3)} \odot \mathbf{C}_{k}^{(1)} \big) = Q, \label{equ:Uniqueness_BeamSquint}
		\end{cases}
	\end{align}
	with $K_{3} + L_{3} = L + 1$, the CPD of $\bm{\mathcal{X}}_{k}$ is said to be unique. 
	In the generic case, \eqref{equ:Uniqueness_BeamSquint} becomes
	\begin{align}
		\min \big( N_{d}(K_{3}-1),   M_\mathrm{re} L_{3} \big) \geq Q. \label{equ:DW_theoUniqueness}
	\end{align}
\end{proposition}
\begin{proof}
	The proof follows that of Lemma \ref{lemma:Lemma2_Vander} by replacing $\mathbf{B}^{(1)}$, $\mathbf{B}^{(2)}$, and $\mathbf{B}^{(3)}$, with 
	$\mathbf{C}_{k}^{(1)}$, $\mathbf{C}_{k}^{(2)}$, and $\mathbf{C}_{k}^{(3)}$, respectively.
\end{proof}

Proposition \ref{prop:Uniqueness_BeamSquint} provides a sufficient condition to determine the uniqueness of CPD in the case of beam squint effects. 
The uniqueness condition holds true for arbitrary subcarrier $k$ in the frequency domain. 
This indicates that the factor matrices can be recovered using individual subcarriers, and thus avoiding the potential performance degradation caused by the beam squint effects.

\subsection{Estimation of Unkown Parameters}

\begin{algorithm}[!tp]
	\renewcommand{\algorithmicrequire}{\textbf{Input:}}  
	\renewcommand{\algorithmicensure}{\textbf{Output:}}  
	\caption{Proposed Parameter Estimation Algorithm With Beam Squint Effects}
	\label{alg:ParaEst_DW}
	\begin{algorithmic}[1]
		\Require
		The received echo signals $\mathbf{\overline{Y}}_{k,l}$, $\forall l \in \{1, 2, \ldots, L\}$, $\mathbf{P}$.
		
		\For {each $k \in [1,\ldots,K]$}
		\State Compute the factor matrices $\hat{\mathbf{C}}_{k}^{(1)}$, $\hat{\mathbf{C}}_{k}^{(2)}$, and $\hat{\mathbf{C}}_{k}^{(3)}$;
		\For {each $q \in [1,\ldots,Q]$}
		\State Estimate $\hat{\theta}_{q}$ via \eqref{equ:AoA_est_BeamSquint}; \
		\State Estimate $\hat{\phi}_{q}$ via \eqref{equ:AoD_est_BeamSquint}; \
		\State Estimate $\hat{\nu}_{q}$ via \eqref{equ:nu_est_BeamSquint}; \
		\EndFor
		\State Estimate $\hat{\overline{\bm{\beta}}}_{k}$ via \eqref{equ:beta_BeamSquint}; \
		\EndFor
		\State Estimate $\hat{\tau}_{q}$ and $\hat{\beta}_{q}$ via \eqref{equ:tau_est_beamSquint} and \eqref{equ:beta_est_beamSquint}, respectively; \

		\Ensure
		Return estimated parameters $\lbrace \hat{\theta}_{q}, \hat{\tau}_{q}, \hat{\phi}_{q}, \hat{\nu}_{q}, \hat{\beta}_{q} \rbrace_{q=1}^{Q}$.
	\end{algorithmic}
\end{algorithm}

Since the unknown parameters in \eqref{equ:tensor_BeamSquint} are nonlinearly coupled, we can employ the similar operations in \eqref{equ:Mode1Unfolding}-\eqref{equ:bb1_hat} to obtain the estimation of the factor matrices, i.e., $\hat{\mathbf{C}}_{k}^{(1)}$, $\hat{\mathbf{C}}_{k}^{(2)}$, and $\hat{\mathbf{C}}_{k}^{(3)}$. 
After the tensor factorization of \eqref{equ:tensor_BeamSquint}, we turn to
extract the unknown parameters utilizing the estimated factor matrices.
Note that the AoA information of targets is characterized by each column of $\mathbf{C}_{k}^{(1)}$. 
Hence, the AoA can be estimated as
\begin{align} \label{equ:AoA_est_BeamSquint}
	\hat{\theta}_{q} 
	= \arg\max_{\theta} \big| \big(\hat{\mathbf{c}}_{q,k}^{(1)}\big)^{H} \mathbf{a}_{\mathrm{re},k}(\theta) \big|^{2},
\end{align}
where $\hat{\mathbf{c}}_{q,k}^{(1)}$ is the $q$-th column of the estimated factor matrix $\hat{\mathbf{C}}_{k}^{(1)}$. 
Similarly, the AoD $\phi$ can be obtained by
\begin{align} 
	\hat{\phi}_{q} = \arg\max_{\phi} \frac{\big| (\hat{\mathbf{c}}_{q,k}^{(2)})^{H} \mathbf{\overline{P}}^{T} \mathbf{a}_{\mathrm{bs},k}(\phi) \big|^{2} }{\big\| \hat{\mathbf{c}}_{q,k}^{(2)} \big\|_{2}^{2} \big\|\mathbf{\overline{P}}^{T} \mathbf{a}_{\mathrm{bs},k}(\phi)\big\|_{2}^{2}}, \label{equ:AoD_est_BeamSquint}
\end{align}
where $\hat{\mathbf{c}}_{q,k}^{(2)}$ is the $q$-th column of $\hat{\mathbf{C}}_{k}^{(2)}$. 
We now discuss how to estimate the Doppler shift $\nu_{q}$ from the estimated factor matrix $\hat{\mathbf{C}}_{k}^{(3)}$. 
Note that since $\hat{\mathbf{C}}_{k}^{(3)}$ exhibits the Vandermonde nature, the Doppler shift can be directly estimated via
\begin{align}
	\hat{\nu}_{q} = \frac{1}{2 \pi T_{\mathrm{sym}}} \angle \hat{z}_{\nu,q,k}, \label{equ:nu_est_BeamSquint}
\end{align}
where $\hat{z}_{\nu,q,k}$ is the estimated generator along the $q$-th column of $\hat{\mathbf{C}}_{k}^{(3)}$.
By constructing $\hat{\mathbf{A}}_{\mathrm{ds}}$, $\hat{\mathbf{C}}_{k}^{(2)}$, and $\hat{\mathbf{C}}_{k}^{(1)}$ based on the estimated parameters $\lbrace \hat{\nu}_{q} \rbrace_{q=1}^{Q}$, $\lbrace \hat{\phi}_{q} \rbrace_{q=1}^{Q}$, and $\lbrace \hat{\theta}_{q} \rbrace_{q=1}^{Q}$, respectively, we can obtain 
\begin{align} 
	\hat{\overline{\bm{\beta}}}_{k} = \big(\hat{\mathbf{A}}_{\mathrm{ds}} \odot  ( \hat{\mathbf{C}}_{k}^{(2)} \odot \hat{\mathbf{C}}_{k}^{(1)} ) \big)^{\dagger} \mathrm{vec} \big(\mathbf{Y}_{(3),k}^{T}\big), \label{equ:beta_BeamSquint}
\end{align}
where $\mathbf{Y}_{(3),k}$ is the mode-3 unfolding of $\bm{\overline{\mathcal{Y}}}_{k}$, the coefficients $\hat{\overline{\bm{\beta}}}_{k}$ is the estimate of $\overline{\bm{\beta}}_{k} = [\overline{\beta}_{1,k}, \ldots, \overline{\beta}_{Q,k}]^T$.

To this end, the explicit estimation of reflection coefficients and time delays has not been obtained. 
To address this issue, it is necessary to combine the coefficients of multiple subcarriers. 
From \eqref{equ:ReceivedSignalMatrix_DW}, we have
\begin{align}
\Big[\overline{\bm{\beta}}_{1}, \ldots, \overline{\bm{\beta}}_{K}\Big] =  \mathrm{diag}(\beta_{1}, \beta_{2}, \ldots, \beta_{Q}) \mathbf{A}_{\mathrm{td}}^{T},
\end{align}
where the $q$-th column of $\mathbf{A}_{\mathrm{td}}$ is constructed from $\mathbf{a}_{\mathrm{td}}(\tau_{q})$.
With the estimated $\hat{\overline{\bm{\beta}}}_{k}$ of $K$ subcarriers, we can estimate the time delay as
\begin{align} \label{equ:tau_est_beamSquint}
\hat{\tau}_{q} & = \arg\max_{\tau_{q}} \frac{\Big|\hat{\overline{\bm{\beta}}}_{q}^{H} \mathbf{a}_{\mathrm{td}}(\tau_{q}) \Big|^{2} }{\left\| \hat{\overline{\bm{\beta}}}_{q} \right\|_{2}^{2} \left\|\mathbf{a}_{\mathrm{td}}(\tau_{q})\right\|_{2}^{2}},
\end{align}
where $\hat{\overline{\bm{\beta}}}_{q}$ extracts the $q$-th row elements of the matrix  $\Big[\hat{\overline{\bm{\beta}}}_{1}, \ldots, \hat{\overline{\bm{\beta}}}_{K}\Big] \in \mathbb{C}^{Q \times K}$. Then, the reflection coefficient is obtained by
\begin{align} \label{equ:beta_est_beamSquint}
\hat{\beta}_{q} = 
e^{\mathrm{j} 2 \pi \hat{\tau}_{q} f_{s} \frac{k}{K_{0}}}
\hat{\overline{\beta}}_{q,k}.
\end{align}

We summarize the proposed parameter estimation method in Algorithm \ref{alg:ParaEst_DW} for the case of beam squint effects. 
Meanwhile, the UE can estimate the multipath channel parameters $\lbrace \hat{\theta}_{p}, \hat{\tau}_{p}, \hat{\phi}_{p}, \hat{\nu}_{p}, \hat{\alpha}_{p} \rbrace_{p=1}^{P}$ under the same process and reconstruct the entire channel matrix.{\footnote{The channel training structure in \eqref{equ:training_BeamSquint} and the Doppler approximation in \eqref{equ:ReceivedSignalMatrix_N1_OFDM} should be known in advance at the UE side for facilitating the tensor-based formulation and estimation.}}

\section{Simulation Results}
\label{Sec:Simulation}

In this section, we conduct the numerical simulation to investigate the performance of the proposed unified channel and target parameter estimation schemes.  
We consider a BS equipped with $M_{\mathrm{bs}} = 64$ transmit antennas, $M_{\mathrm{re}} = 8$ antennas, and $M_{\mathrm{ue}} = 8$ antennas at the UE. The carrier frequency is $f_{c} = 28$ GHz and the bandwidth is $f_{s} = 100$ MHz with $K_{0}=128$.  
The elements of training precoder $\mathbf{P}$ are uniformly drawn from a unit circle \cite{zhou2017low}. 
The CP duration in one OFDM symbol is set to $T_{\mathrm{cp}} = \frac{1}{2\Delta f}$ and the time delay is uniformly generated between $0$ and $T_{\mathrm{cp}}$.
The BS is detecting $Q = 4$ targets and the channel experiences $L = 4$ multipaths, where the reflection coefficient and path gain are generated according to $\mathcal{CN}(0,1)$.
The AoAs and AoDs for either targets or channels are randomly drawn from $[-\pi/3, \pi/3]$. 
The velocity of each target $V_{q}$ is uniformly distributed in 
$[-V_{\mathrm{max}}, V_{\mathrm{max}}]$ with the maximum velocity $V_{\mathrm{max}} = 30$ m/s, and the maximum Doppler shift of the channel is $\nu_{\max} = 2.8$ kHz.
The root mean square error (RMSE) is adopted for examining the target parameters estimation performance \cite{Zhang2020Joint}
\begin{equation}
\mathrm{RMSE}(\chi) = \sqrt{\frac{1}{Q} \sum_{q=1}^{Q} \big(\hat{\chi}_{q} - \chi_{q} \big)^2 },
\end{equation}
where $\hat{\chi}_{q}$ and $\chi_{q}$ denote the estimated and the corresponding true parameters, respectively, $\chi_{q} \in \lbrace \theta_{q}, \phi_{q}, \tau_{q}, \nu_{q} \rbrace$. The normalized MSE (NMSE) at the $n$-th OFDM symbol is calculated as 
\begin{equation}
\mathrm{NMSE}(n) = \frac{\sum_{k=1}^{K} \big\lVert \mathbf{H}_{n,k} - \mathbf{\hat{H}}_{n,k} \big\rVert_{F}^{2}}{\sum_{k=1}^{K} \big\lVert \mathbf{H}_{n,k} \big\rVert_{F}^{2}},
\end{equation}
and we adopt $\mathrm{NMSE} = \mathrm{NMSE}(N)$ to measure the channel estimation performance.

\begin{figure*}[htbp]
	\subfigure[]{
		\begin{minipage}[t]{0.4\linewidth}
			\centering
			\includegraphics[width=3.3in]{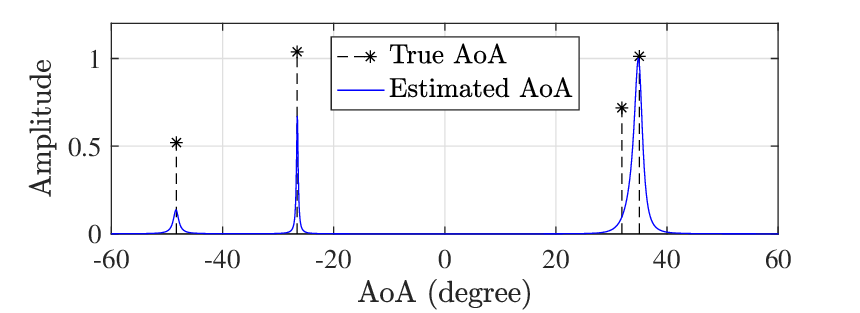}
		\end{minipage}%
	}%
	\subfigure[]{
		\begin{minipage}[t]{0.7\linewidth}
			\centering
			\includegraphics[width=3.3in]{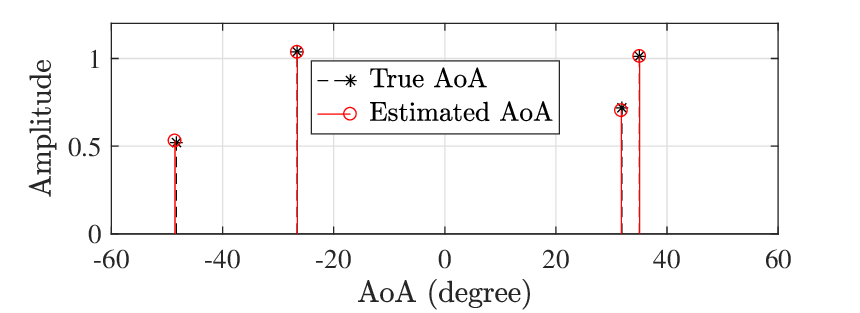}
		\end{minipage}%
	}%
	\quad
	\subfigure[]{
	\begin{minipage}[t]{0.4\linewidth}
		\centering
		\includegraphics[width=3.0in]{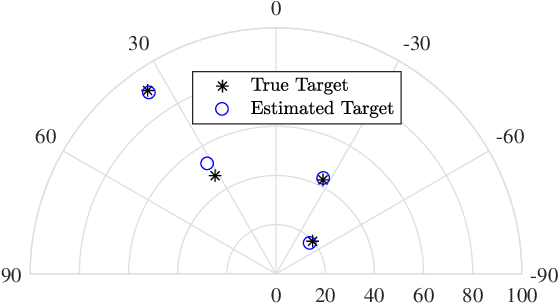}
	\end{minipage}%
    }%
    \subfigure[]{
	\begin{minipage}[t]{0.7\linewidth}
		\centering
		\includegraphics[width=3.0in]{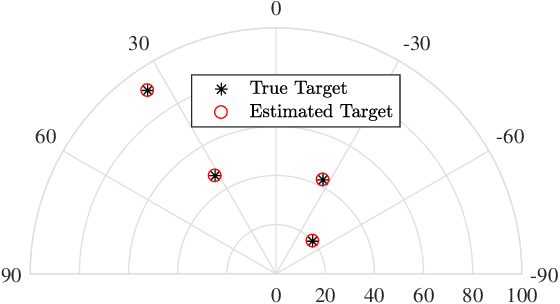}
	\end{minipage}%
	}%
    \quad 
	\subfigure[]{
	\begin{minipage}[t]{0.4\linewidth}
		\centering
		\includegraphics[width=3.3in]{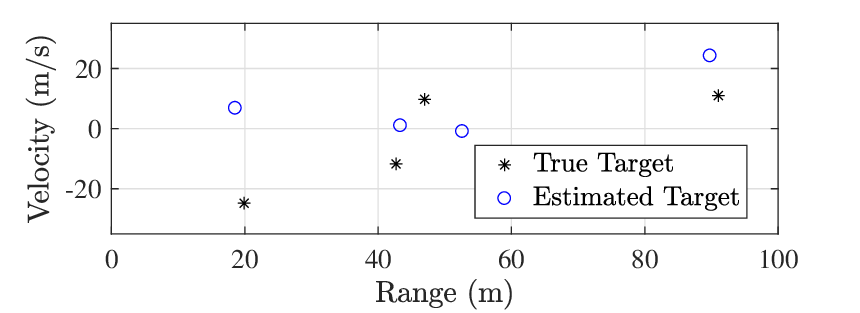}
	\end{minipage}%
	}%
	\subfigure[]{
	\begin{minipage}[t]{0.7\linewidth}
		\centering
		\includegraphics[width=3.3in]{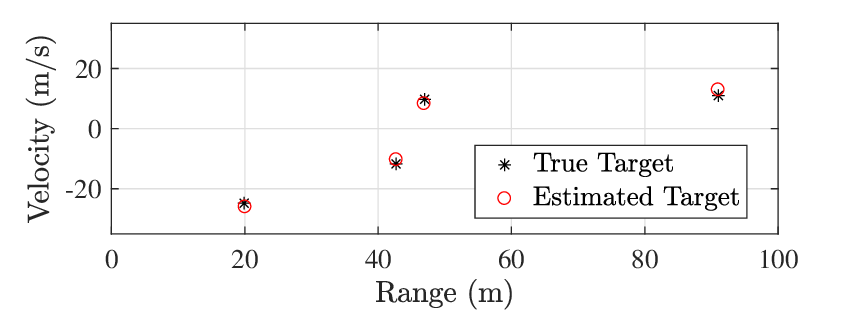}
	\end{minipage}%
	}%
	\centering
	\caption{Target parameter estimation performance for the benchmark algorithm (left column) and the proposed Algorithm 1 (right column), where SNR $=10$ dB, $N = 16$, $K = 16$.  
	The first row shows the AoA estimation; 
	The second row shows the AoD and range estimation (polar axis); 
	The third row shows the range and velocity estimation (x-axis range, y-axis velocity).
	}
\end{figure*}

We investigate the target parameter estimation performance of the proposed Algorithm 1 in Fig. 2, where the SNR
$\left\|\mathbf{Y}_{k} - \mathbf{N}_{k} \right\|_{F}^2 / \left\|\mathbf{N}_{k}\right\|_{F}^2$ is set to $=10$ dB, $N = 16$, $K = 16$, $I_{\mathrm{iter}} = 30$. 
The benchmark algorithm in \cite{liu2020joint} is taken into consideration, where the AoA/AoD and range/velocity are estimated via the MUSIC and matched-filtering (MF) methods, respectively.  
For the proposed algorithm, it can be observed that all the target paramaters (angles, ranges, and velocities) can be accurately estimated only using few time-frequency training resources.
In contrast, only part of the target parameters can be accurately recovered by the benchmark algorithm. 
For instance, from Fig. 2(a), the two adjacent targets located at azimuth angles $31.87^{\circ}$ and $35.01^{\circ}$ fail to be identified by the MUSIC method, which is due to the limited aperture of the receive antenna array. 
From Fig. 2(e), although the range parameters can be effectively estimated, there is large error for target velocity estimation. 
This is because the number of training OFDM symbols used for target sensing is fewer than that of transmit antennas in massive MIMO-ISAC, resulting in short observation duration and low Doppler resolution.

\begin{figure}[!tp]
	\centering
	\includegraphics[width=3.2in]
	{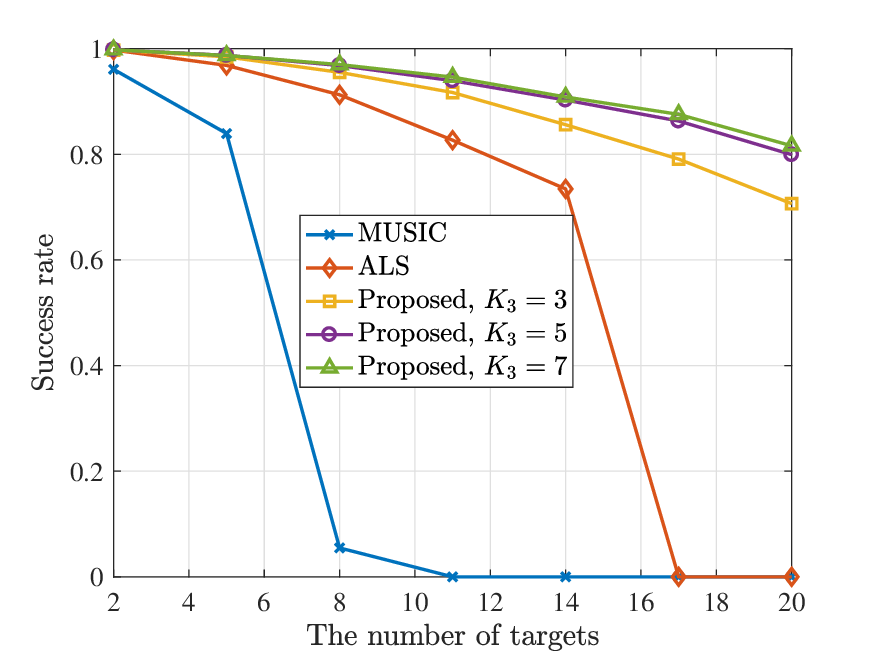}
	\caption{The success rate of correctly distinguishing the targets.}
	\label{Fig_SuccessRate_NumTargets}
\end{figure}

Fig. \ref{Fig_SuccessRate_NumTargets} depicts the success rate of correctly distinguishing the targets, where the ALS-based method in \cite{zhou2017low,park2019spatial} and MUSIC method are taken into comparison. 
A trial is considered successful if the AoA estimation error satisfies $|\sin \theta_{q} - \sin \hat{\theta}_{q} | \leq \frac{1}{2 M_{\mathrm{re}}} $ \cite{Mamandipoor2016Newtonized}.
The simulation parameters are set to SNR $=15$ dB, $N = 16$, $K = 16$. 
From Fig. \ref{Fig_SuccessRate_NumTargets}, it can be observed that the success rate of the proposed algorithm is considerably higher than that of benchmark methods.  
The success rate of the MUSIC method decreases to zero when the number of targets to be detected exceeds 11. 
The ALS-based method can detect no more than 17 targets, while the success rate of the proposed algorithm is still more than $79.1 \%$.
The significant enhancement of the success rate is attributed to exploiting the inherent Vandermonde structure of factor matrices in \eqref{equ:Mode1Unfolding}-\eqref{equ:bb1_hat}, which experimentally verifies the analysis in Lemma \ref{lemma:Lemma2_Vander}.

\begin{figure*}[htbp]
	\centering 
	\subfigure[AoD estimation.]{
		\begin{minipage}[t]{0.4\linewidth}
			\centering
			\includegraphics[width=3.3in]{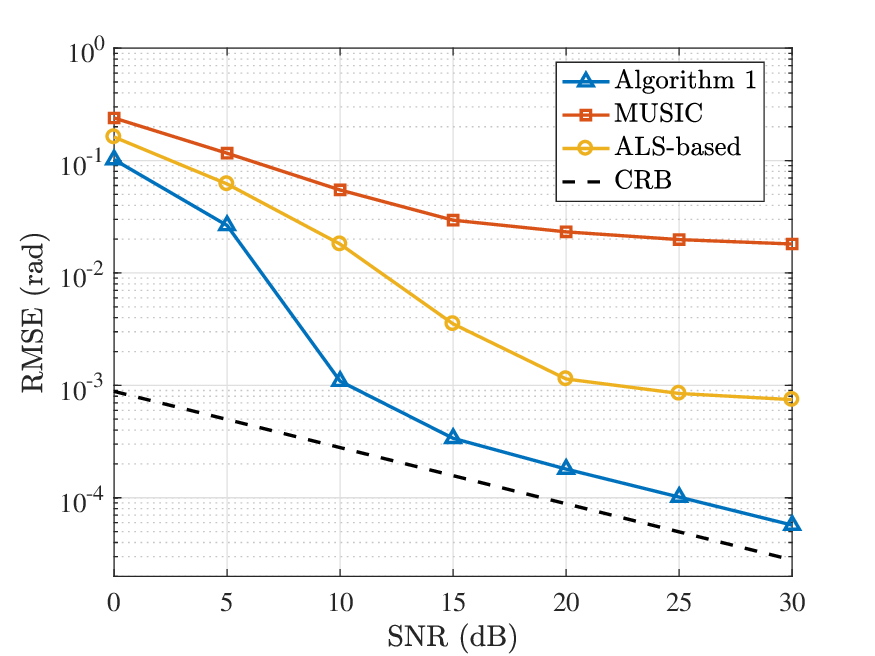}
			\label{Fig_VaryingSNR_AoD}
		\end{minipage}%
	}%
	\subfigure[AoA estimation.]{
		\begin{minipage}[t]{0.7\linewidth}
			\centering
			\includegraphics[width=3.3in]{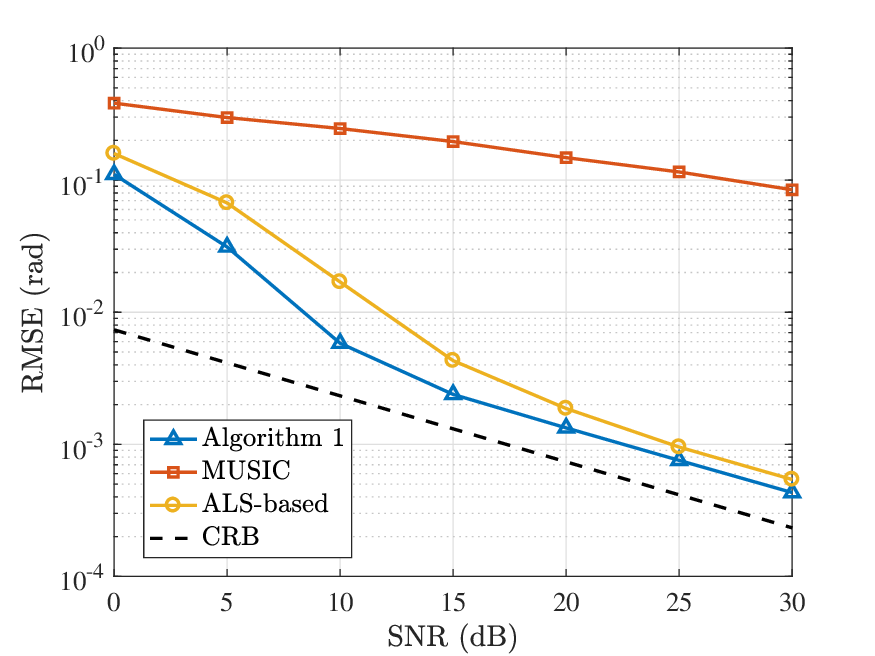}
			\label{Fig_VaryingSNR_AoA}
		\end{minipage}%
	}%
	\centering
	\caption{The RMSE performance of the AoD and AoA estimation versus SNR.}
\end{figure*}

Fig. 4 depicts the RMSE performance of the AoD and AoA estimation versus SNR for the proposed algorithm, which is compared with the ALS-based method, the MUSIC method, and the Cramér-Rao bound (CRB) derived in \cite{zhou2017low}.
The estimation results are averaged over $95\%$ of the realizations to ignore the effect of outliers \cite{Myers2019Message}.
As can be observed from Fig. 4, the AoD estimation performance of the proposed Algorithm 1 becomes better as SNR increases. 
The ALS-based method and MUSIC method encounter the RMSE bottlenecks as the SNR increases, since they fail to address the coupling problem between the AoD and Doppler shift. 
Also, the proposed algorithm achieves better AoA estimation performance than the conventional ALS-based method. 
The improvement is attributed to the tensor formulation of the target parameter estimation in \eqref{equ:TargetParameterEstimation}, where the Vandermonde structure of the factor matrix is exploited.
In contrast, due to the limited training symbols, the MUSIC method fails to return a satisfactory performance. 
Besides, from Fig. 4, we can observe that the AoD and AoA estimation accuracy attained by the proposed Algorithm 1 gradually approaches their corresponding CRBs, especially in the high SNR regime. 
These results validate the superior performance of the proposed algorithm in terms of angle estimation.

\begin{figure*}[htbp]
	\subfigure[Range estimation.]{
		\begin{minipage}[t]{0.4\linewidth}
			\centering
			\includegraphics[width=3.3in]{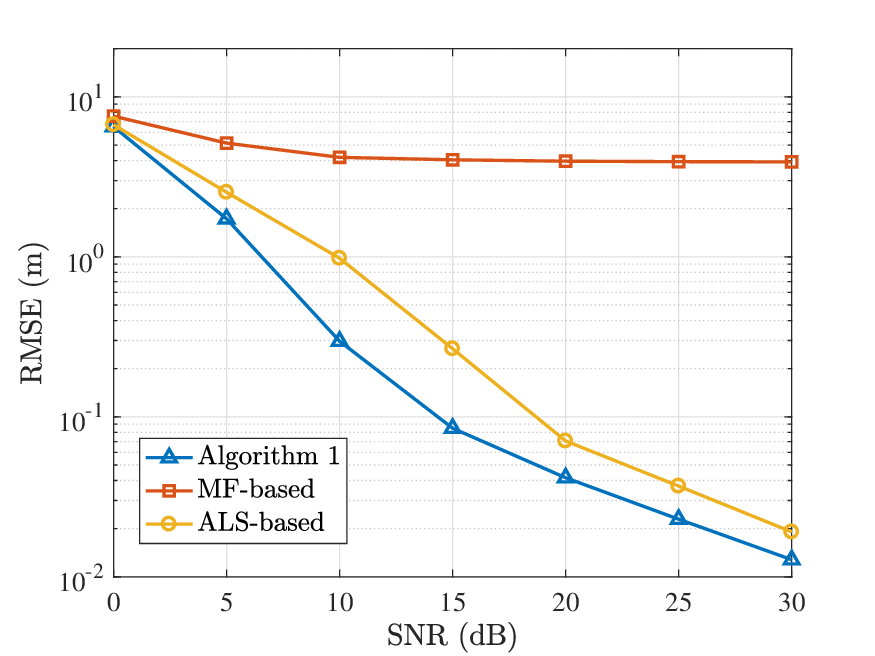}
			\label{Fig_VaryingSNR_Range}
		\end{minipage}%
	}%
	\subfigure[Velocity estimation.]{
		\begin{minipage}[t]{0.7\linewidth}
			\centering
			\includegraphics[width=3.3in]{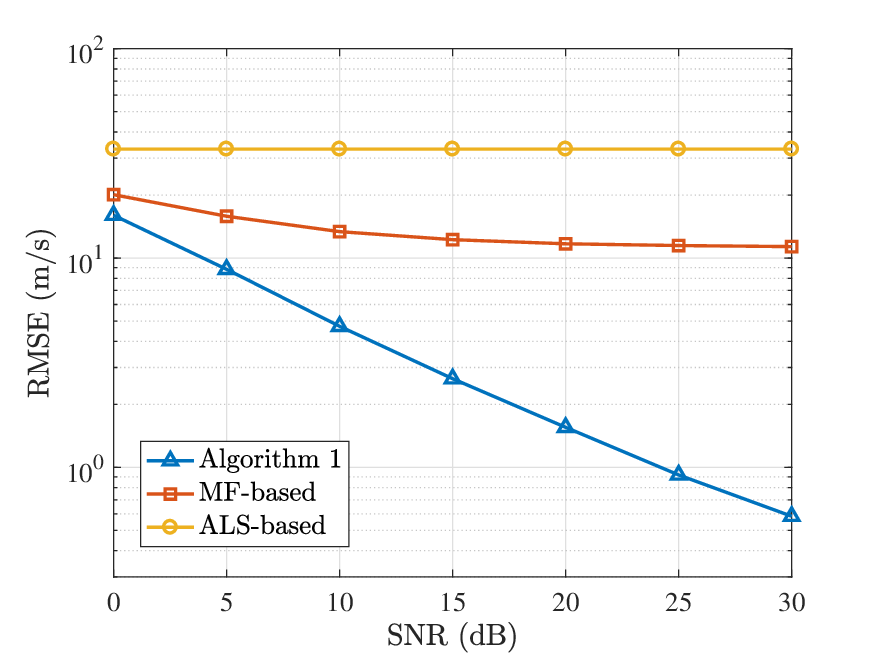}
			\label{Fig_VaryingSNR_Velocity}
		\end{minipage}%
	}%
	\centering
	\caption{The RMSE performance of the range and velocity estimation versus SNR.
	}
\end{figure*}

Fig. 5 plots the RMSE performance of the range and velocity estimation as a function of SNR under the same simulation setting. 
Compared to the conventional MF-based method, the range estimation performance for both the proposed algorithm and the ALS-based method gradually improves as SNR increases. 
The more accurate range estimation is because the time delay caused by multiple targets can be independently estimated from the factor matrices in tensor-based methods. 
In Fig. \ref{Fig_VaryingSNR_Velocity}, the conventional ALS-based method yields the worst RMSE performance of the velocity estimation, since it cannot recover the coupled Doppler shift and angle information from the factor matrix $\hat{\mathbf{B}}^{(2)}$.  
The MF-based method achieves improved velocity estimation performance but gradually approaches a floor accuracy. 
The proposed Algorithm 1 can achieve the refined velocity estimation performance via the alternative optimation and the Doppler compensation. 
Note that the RMSE of range and velocity estimation is relatively higher than that of angle estimation. The reason is that the range, velocity, and angle parameters correspond to the resources in the frequency, time, and space domains, respectively, and the occupied resources for these parameters are different.

\begin{figure}[!tp]
	\centering
	\includegraphics[width=3.2in]
	{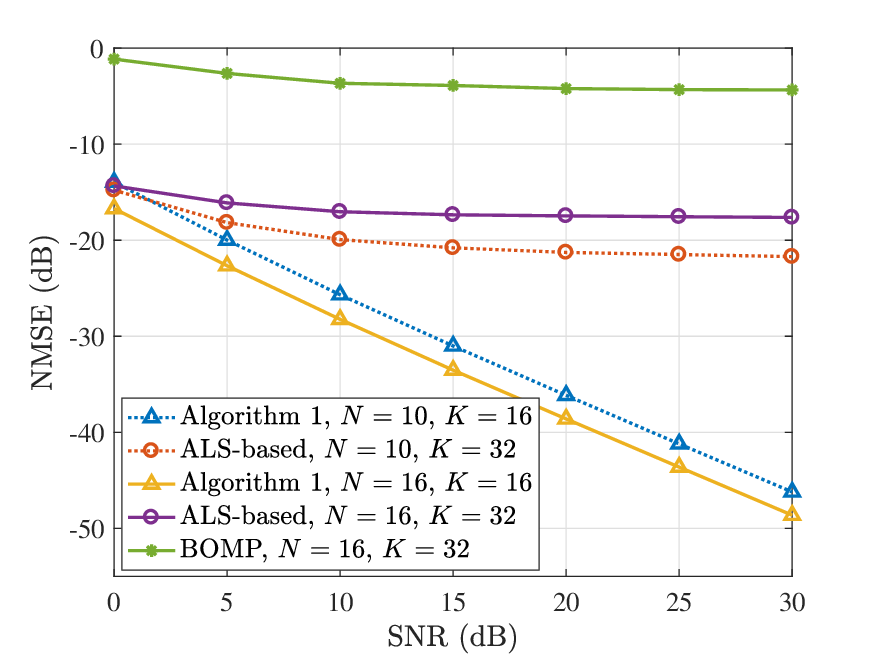}
	\caption{The NMSE performance of channel estimation as a function of SNR.}
	\label{Fig_VaryingSNR_H}
\end{figure}

Fig. \ref{Fig_VaryingSNR_H} shows the NMSE performance of the channel estimation as a function of SNR for the proposed algorithm, where the ALS-based method and the block orthogonal matching pursuit (BOMP) algorithm in \cite{qin2018time} are taken into account.   
The number of training subcarriers for the proposed algorithm is $K = 16$, while the benchmark algorithms employ $K=32$ training subcarriers. 
It can be observed that the proposed algorithm can achieve superior performance with less training overhead.  
Due to the ignorance of Doppler effects, the ALS-based method encounters a performance bottleneck as the SNR increases.
Although more training symbols are used, the ALS-based method cannot return the improved channel estimation performance since more phase distortion is accumulated in the presence of the Doppler shifts. 
On the contrary, the proposed algorithm yields the improved estimation performance as number of training symbols increases. 
This is attributed to the joint estimation of Doppler shift and AoD in \eqref{equ:JointAoDandDopplerProblem}, where the path gain estimation can be improved by the effective Doppler compensation.
Note that the BOMP algorithm does not attain satisfactory channel estimation performance in this scenario. 
The possible reason is that the number of training subcarriers and symbols is small, whereas the number of unknown nonzero elements in high-dimensional massive MIMO channels is exceedingly large. 
This incurs the high correlation between columns of the dictionary matrix, resulting in the failure of channel support set reconstruction.

\renewcommand\arraystretch{1.4}
\begin{table}[!tp]
	\centering
	\caption{Comparison of CPU runtime (s) using different algorithms for channel and target parameters estimation} 
	\begin{tabular}{|c|c|c|c|c|c|}
		\hline
		Algorithms & BOMP & MUSIC & MF & ALS & Proposed \\
		\hline
		Factor matrix & $\diagdown$ & $\diagdown$ & $\diagdown$ & 0.1915 & 0.0037 \\
		\hline
		AoA & $\diagdown$ & 0.0049 & $\diagdown$ & 0.0027 & 0.0027 \\
		\hline
		AoD & $\diagdown$ & 0.0153 & $\diagdown$ & 0.0058 & \multirow{2}{*}{0.7742}  \\
		\cline{1-1} \cline{2-2} \cline{3-3} \cline{4-4}  \cline{5-5}
		Velocity & $\diagdown$ & $\diagdown$ & 0.0024 & $\diagdown$ &  \\
		\hline
		Range & $\diagdown$ & $\diagdown$ & 0.0007 & 0.0067 & 0.0068 \\
		\hline
		Channel & 0.2366 & $\diagdown$ & $\diagdown$ & 0.2068 & 0.5986 \\
		\hline
	\end{tabular}
	\label{Table:CPUruningtime}
\end{table}

\begin{figure*}[htbp]
	\subfigure[AoD estimation.]{
		\begin{minipage}[t]{0.4\linewidth}
			\centering
			\includegraphics[width=3.3in]{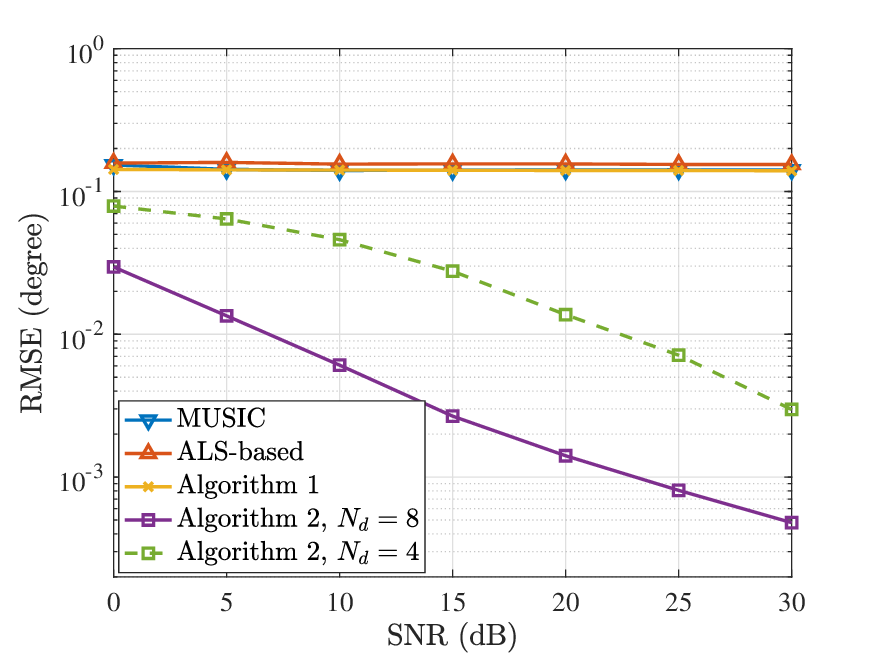}
			\label{Fig_DW_VaryingSNR_AoD}
		\end{minipage}%
	}%
	\subfigure[AoA estimation.]{
		\begin{minipage}[t]{0.7\linewidth}
			\centering
			\includegraphics[width=3.3in]{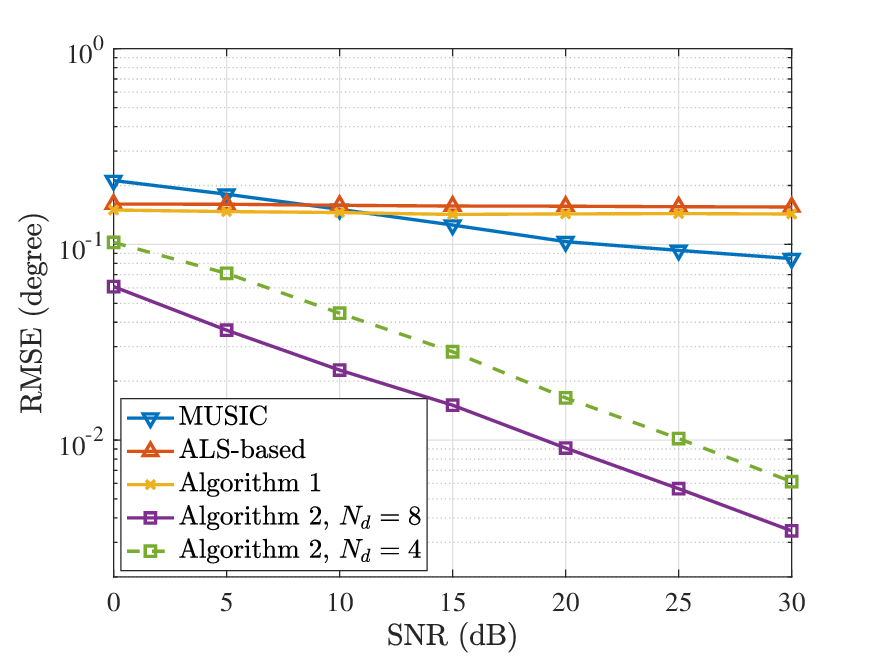}
			\label{Fig_DW_VaryingSNR_AoA}
		\end{minipage}%
	}%
	\centering
	\caption{The RMSE performance of the AoD and AoA estimation versus SNR with significant beam squint effects.
	}
\end{figure*}

Next, we compare the CPU runtime of the proposed algorithm with benchmark algorithms for channel and target parameters estimation in Table \ref{Table:CPUruningtime}.
Compared to the MUSIC and MF methods, the proposed algorithm maintains competitive computational complexity in terms of AoA and range estimation, but has much higher complexity in terms of AoD and velocity estimation. 
Also, against to the ALS-based method, the proposed algorithm shows substantially lower computational complexity for factor matrix estimation and similar computational complexity for AoA and range estimation. 
However, for AoD and velocity estimation, the proposed algorithm exhibits much higher CPU runtime. 
This is because the proposed algorithm requires the iterative estimation process for addressing the coupling between AoD and Doppler shift parameters, resulting in the higher CPU runtime. 
In terms of channel estimation, although the CPU runtime of the proposed algorithm is more than double that of the BOMP and ALS-based method, this leads to a significant improvement in estimation performance, as shown in Fig. \ref{Fig_VaryingSNR_H}. 
Thus, when compared to the benchmark algorithms, the computational complexity of the proposed algorithm is competitive for AoA and range estimation of targets, though it is higher for channel estimation and the joint AoD and velocity estimation of targets.

We next investigate the parameter and channel estimation performance of the proposed algorithm in the case of significant beam squint effects. 
In this case, we set $f_{s} = 1$ GHz, $T_{\mathrm{cp}} = 1/ \Delta f$, $K = K_{0} = 128$, $N = 64$, 
$N_{d} = 8$,  $L = 8$, $K_3 = 5$.
Since the parameters are estimated using individual subcarriers for the proposed Algorithm 2, we select the minimum RMSE and NMSE among the training subcarriers as the evaluation metric.

Fig. 7 illustrates the RMSE performance of angle parameters as a function of SNR, where the MUSIC method, the ALS-based method, and the proposed Algorithm 1 are compared.
From Fig. \ref{Fig_DW_VaryingSNR_AoD}, we can see that the proposed Algorithm 2 can attain the improved AoD estimation performance as SNR gradually increases, while the conventional methods as well as Algorithm 1 cannot offer satisfactory RMSE performance due to the interference caused by beam squint effects.
A similar tendency can be observed for the AoA estimation in Fig. \ref{Fig_DW_VaryingSNR_AoA}. 
This highlights the advantage of the proposed segment-based training pattern, such that the parameter estimation in the whole bandwidth can be decoupled and performed via the individual subcarriers.

\begin{figure}[!tp]
	\centering
	\includegraphics[width=3.2in]
	{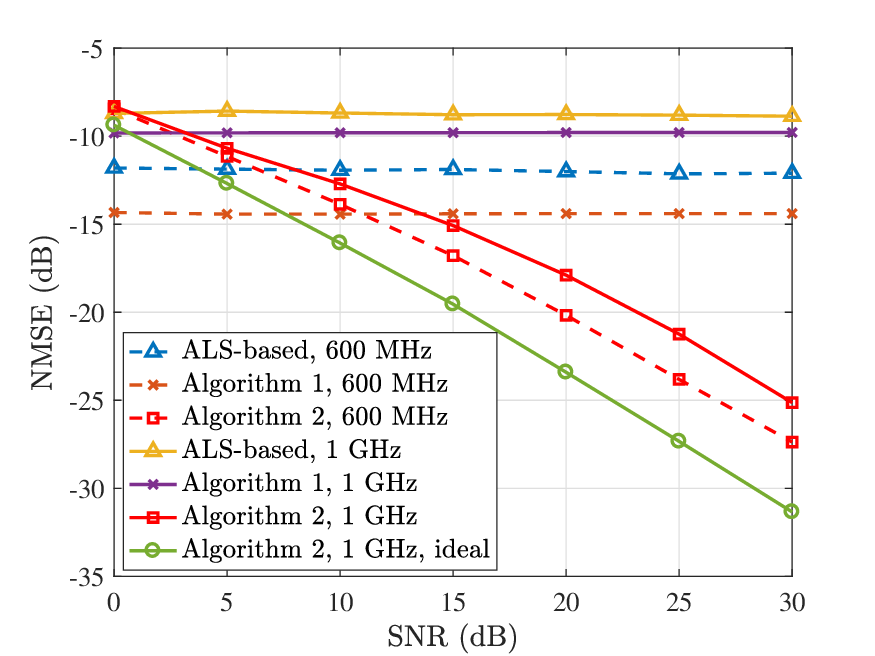}
	\caption{The NMSE performance of channel estimation as a function of SNR with different bandwidths, illustrating the significant beam squint effects.}
	\label{Fig_DW_VaryingSNR_H}
\end{figure}

Fig. \ref{Fig_DW_VaryingSNR_H} depicts the NMSE performance of the channel estimation as a function of SNR, where the system bandwidth is set to $600$ MHz and $1$ GHz, respectively.  
As can be observed, the proposed Algorithm 2 yields the better channel estimation accuracy and gradually improves as the SNR increases. 
Meanwhile, the ideal case without the Doppler approximation error is also considered. In this case, there exists a reasonable performance gap between the proposed Algorithm 2 under the ideal case and the Doppler approximation errors.
In contrast, the performance of the conventional ALS-based method and Algorithm 1 hits the NMSE floor as the SNR increases, since there is serious interference caused by beam squint effects. 
When the bandwidth decreases to $600$ MHz, we can see that the NMSE performance of benchmark schemes is improved, which is due to the reduction of the beam squint effects.

\section{Conclusion}
In this paper, a unified tensor framework for channel and target parameter estimation was developed for massive MIMO-ISAC systems, considering the cases with negligible and severe beam squint effects, respectively. 
By parameterizing the high-dimensional communication channel into a small number of physical parameters, we established a strong connection between estimating the wireless channel and the target parameters from the perspective of angular, delay, and Doppler dimensions. 
We proposed to exploit the same time-frequency resources for simultaneously channel and target sensing, whose parameters estimation can be both formulated as two structured tensor decomposition problems, complying with a unified CPD model. 
Based on this model, we analyzed the uniqueness condition of the formulated tensor decomposition problem and leveraged the Vandermonde structure of the factor matrices to enable the enhanced estimation of channel and target parameters including AoA, AoD, time delay, Doppler shift, and coefficients. 
Numerical results verified our theoretical analysis and the superiority of the proposed unified channel and target parameter estimation approaches in terms of the number of resolvable targets, sensing resolution, training overhead reduction, and channel estimation accuracy.

\vspace{1em}



%

\bibliographystyle{IEEEtran}

\end{document}